\documentclass[11pt]{article}
\usepackage[T1]{fontenc}
\usepackage{amsmath,amsxtra,amssymb,latexsym,amscd,amsthm,pb-diagram,fancyhdr,euscript}
\usepackage{amsthm}
\usepackage{indentfirst}
\usepackage{epsfig}
\usepackage{mathrsfs}
\usepackage{slashed}
\usepackage{hyperref}
\usepackage{float}
\hypersetup{
   colorlinks,
    citecolor=blue,
    filecolor=black,
    linkcolor=blue,
    urlcolor=magenta
}
\hypersetup{linktocpage}

\renewcommand{\d}{\mathrm{d}}

\newcommand{\scri}{{\mathscr I}}

\newcommand{\thorn}{\mbox{\th}}
\newcommand{\hnabla}{\hat{\nabla}}
\newcommand{\hook}{{\setlength{\unitlength}{11pt}   
                   \begin{picture}(.833,.8)
                   \put(.15,.08){\line(1,0){.35}}
                   \put(.5,.08){\line(0,1){.5}}
                   \end{picture}}}

\newtheorem{theorem}{Theorem}

\newtheorem{remark}{Remark}

\newtheorem{defn}{Definition}[section]
\newtheorem{thm}{Theorem}[section]

\newtheorem{cor}{Corollary}[section]

\begin{document}
\mbox{} \thispagestyle{empty}

\begin{center}
\bf{\Huge Cauchy and Goursat problems for the generalized spin zero rest-mass equations on Minkowski spacetime} \\

\vspace{0.1in}

{PHAM Truong Xuan\footnote{Faculty of Information Technology, Department of Mathematics, Thuyloi university, Khoa Cong nghe Thong tin, Bo mon Toan, Dai hoc Thuy loi, 175 Tay Son, Dong Da, Ha Noi, Viet Nam.\\
Email~: xuanpt@tlu.edu.vn and phamtruongxuan.k5@gmail.com}}
\end{center}

{\bf Abstract.} In this paper, we study the Cauchy and Goursat problems of the spin-$n/2$ zero rest-mass equations on Minkowski spacetime by using the conformal geometric method. In our strategy, we prove the wellposedness of the Cauchy problem in Einstein's cylinder. Then we establish pointwise decays of the fields and prove the energy equalities of the conformal fields between the null conformal boundaries $\scri^\pm$ and the hypersurface $\Sigma_0=\left\{ t=0 \right\}$. Finally, we prove the wellposedness of the Goursat problem in the partial conformal compactification by using the energy equalities and the generalisation of H\"ormander's result. 



\tableofcontents

\section{Introduction}
The spin-$n/2$ zero rest-mass fields were studied since 1960's in the works of Sachs \cite{Sa61} and Penrose \cite{Pe1964,Pe1965}. The authors discovered the ''peeling-off'' properties of the fields along the outgoing null geodesic lines in Minkowski spacetime. Then, the ''peeling-off'' properties have been extensively to study in \cite{ChruDe2002,CoScho2003,KaNi2003,Shu}. Specifically, in the asymptotic flat spacetimes Mason and Nicolas \cite{MaNi2009,MaNi2012} constructed the optimal space of the initial data which guarantees the ''peeling-off'' property in Schwarzschild spacetime for the scalar, Dirac and Maxwell fields, i.e, spin-$0$, spin-$1/2$ and spin-$1$ fields respectively. Recently, Nicolas and Pham \cite{NiXu2019,Xuan2020} have extended the works of Mason and Nicolas for the wave and Dirac equations in Kerr spacetime.

The pointwise decays (also called Price's law) of the spin-$n/2$ zero rest-mass fields in Minkowski spacetime were established by Andersson et al. in \cite{ABJ} by analyzing Hertz potentials. The authors obtained the existence and pointwise estimates for the Hertz potentials using a weighted estimate for the spin-wave equation, then applied to give weighted estimates for the solutions of the spin-$n/2$ zero rest-mass field equations. Specifically, the pointwise decays for the Maxwell field on the black hole spacetimes such as Schwarzschild and Kerr spacetimes were studied by Tataru et al. \cite{MeTaTo2017} and the results for Dirac field was obtained by Smoller and Xie \cite{SmolXi}. Recently, the almost Price's law for Dirac and Maxwell fields has been studied in Schwarzschild and very slowly Kerr spacetimes in the works of Ma \cite{Ma2020,Ma2021}.

Another interesting aspect of these fields is the local integral formula which was established initially in Minkowski spacetime by Penrose \cite{Pe63}. Then, Joudioux \cite{Jo2011} extended the formula on the general curved spacetimes. The local integral formula gives the solution of the Goursat problem in the region near timelike infinity $i^\pm$ of Minkowski spacetime.

Concerning the Goursat problem of the spin non-zero field equations, Mason and Nicolas established the wellposedness of the scalar wave, Dirac and Maxwell equations in the asymptotically simple spacetimes in \cite{MaNi2004}. By using these results, they constructed the conformal scattering operators, i.e, the geometric scattering operators for these field equations in the asymptotic simple spacetimes. After that, the Goursat problem for the spin field equations in the asymptotic flat spacetime is also established in some recent works. In particular, the Goursat problem for the scalar wave equation has been studied in Schwarzschild spacetime by Nicolas \cite{Ni2016} and the ones for Dirac and Maxwell equations in the Reissner-Nordstr\"om-de Sitter spacetime have been treated by Mokdad \cite{Mo2019,Mo2021}. The authors in \cite{MaNi2004,Ni2016,Mo2019,Mo2021} have used the geometric methods to establish the wellposedness of the Goursat problem. In detail, they combined the vector field method (energy estimates) and the generalisation of H\"ormander's result (see \cite{Ho1990,Ni2006}) to obtain the full solution to the problem. The Goursat problem is an important step to construct the conformal scattering theory for the field equations in the asymptotic simple and flat spacetimes (in detail see \cite{MaNi2004,Ni2016} for the construction of the theory). In related work, the wellposedness of the Goursat problem and conformal scattering theories for the Regge–Wheeler and Zerilli equations in Schwarzschild spacetime have been obtained by Pham \cite{Xuan2021}. 

In the present paper, we study the Cauchy and Goursat problems for the spin-$n/2$ zero rest-mass equations in Minkowski spacetime by using geometric methods. We know that Minkowski spacetime is embedded fully into Einstein's cylinder by Penrose's conformal mapping. We prove the wellposedness of the Cauchy problem of the equations in Einstein's cylinder by using the matrix form of the equations and Leray's theorem for the wellposedness of the global hyperbolic systems (see Theorem \ref{cauchyproblem}). Then, we apply the results to wellposedness of the equations in the full and partial conformal compactification spacetimes. As consequences of the Cauchy problem are that we can define the trace operators on the null hypersurfaces $\scri^\pm$ and then establish the energy equality between the null boundaries $\scri^\pm$ and the Cauchy hypersurface $\Sigma_0=\left\{ t=0 \right\}$ in the full conformal compactification spacetime. Using again the fully conformal compactification of Minkowski spacetime where $i^\pm$ are finite points, we obtain the pointwise decays, i.e, decays in time of all the components of the origin sin-$n/2$ zero rest-mass fields (see Theorem \ref{decay}). We use these pointwise decays to prove the energy equality in the partial conformal compactification, where $i^\pm$ are infinite points (see Theorem \ref{egalite_energies}). 

We develop the methods in \cite{MaNi2004,Mo2019,Ni2016,Xuan2021} to establish the wellposedness of the Goursat problem. In particular, the Goursat problem will be considered in the partial conformal compactification spacetime in two parts. The first one we apply the generalisation of H\"ormander's result to obtain the solution in the future $\mathcal{I}^+(\mathcal{S})$ of the Cauchy hypersurface $\mathcal{S}$ which intersects strictly at the past of the support of initial data. In the second one, we extend the solution of the first part down to the initial hypersurface $\Sigma_0$. This corresponds to prove the wellposedness of the Cauchy problem in the domain $\mathcal{I}^-(\mathcal{S})$ with the initial data which is the restriction of the first step's solution on $\mathcal{S}$. We obtain the solution in the second step by using again the wellposedness of the Cauchy problem in the full conformal compactification and the energy equalities. Finally, the solution of the Goursat problem is a union of the ones obtained in the two parts $\mathcal{I}^+(\mathcal{S})$ and $\mathcal{I}^-(\mathcal{S})$ (see Theorem \ref{Goursatprob}). As a direct consequence of the wellposedness of the Goursat problem and the energy equalities we show that there exists a conformal scattering operator, i.e, a geometric scattering operator which associates the past scattering data on $\scri^-$ to the future scattering data on $\scri^+$ (see Corollary \ref{Con}).

The paper is organized as follows: Section \ref{S2} we recall the geometric setting of Minkowski spacetime which consists of the full and partial conformal compactifications, Section \ref{S3} we describe the spin-$n/2$ zero rest-mass fields and equations in the spin frames, Section \ref{S4} we prove the wellposedness of the Cauchy problem and establish the timelike decays, Section \ref{S5} relies on the energy fluxes of the fields and the proof of the energy equality, Section \ref{S6} we establish the wellposedness of the Goursat problem and construct a conformal scattering operator and finally in Appendix \ref{A} first we recall some formulas of curvature spinors and spinor form of commutators, then we establish the constraint system by using intrinsic space spinor derivatives and prove the existence of a nontrivial solution of this system, we give the generalisation of H\"ormander's result and some calculations which are necessary to prove the wellposedness of the Goursat problem.\\
{\bf Remarks and Notations.}\\
$\bullet$ For the higher spin-$n/2$ zero rest-mass fields ($n\geq 3$) our results valid only in Minkowski spacetime due to in the general curved spacetime the spin-$n/2$ massless equations have only trivial solution.\\
$\bullet$ We use the formalisms of abstract indices, $(n+1)$-component spinor, Newman-Penrose and Geroch-Held-Penrose.\\
$\bullet$ We use the notation $\underbrace{n_an_b...n_c}_{k \; factors} \underbrace{l_dl_e...l_f}_{n-k \; factors}$ to denote the sum of the $C_n^k$ components, where each component has $k$ factors $n_i$ and $n-k$ factors $l_j$ $(i,j \in \left\{ a,b...f \right\} )$.

\section{Geometric setting} \label{S2}
In this section, we recall the conformal structures of Minkowski spacetime (for more details see Penrose \cite{Pe1964,Pe1965} and also Nicolas \cite{Ni2015}) that are basic frameworks to study the wellposedness of Cauchy and Goursat problems. In spherical coordinates $(t,\, r,\, \theta,\, \varphi)$  Minkowski spacetime is $\mathbb{R}^{1+3}$- Lorentzian manifold $\mathbb{M}$ endowed with the metric
\begin{equation}
g = \d t^2 - \d r^2 - r^2 (\d \theta^2 + \sin^2\theta \d \varphi^2).
\end{equation}
The Newman-Penrose tetrad normalization can be chosen as
$$l^a = \frac{1}{\sqrt 2} (\partial_t + \partial_r), \; n^a = \frac{1}{\sqrt 2}(\partial_t - \partial_r), \; m^a = \frac{1}{r\sqrt 2} \left( \partial_\theta + \frac{i}{\sin\theta}\partial_\varphi \right).$$
They are associated with the normalized spin-frame $\left\{ o^A,\iota^A \right\}$ by
$$l^a = o^A o^{A'}, \; n^a = \iota^A \iota^{A'}, \; m^a = o^A \iota^{A'}.$$
The volume form associated with metric $g$ is
$$\mathrm{dVol}^4_g = r^2\sin\theta \d t\d r\d \theta \d \varphi = r^2 \d t \d r \d^2\omega,$$
where $\d^2\omega$ denotes the volume form of the unit $2$-sphere $S^2$.

\subsection{The full conformal compactification}\label{Embedd_Minkowski}
Following the advanced and retarded coordinates: $u = t-r \; , \; v = t+r$ we  put
$$p = \arctan u \; , \; q = \arctan v,$$
$$\tau = p + q = \arctan(t-r) + \arctan(t+r),$$
$$\zeta = q - p = \arctan (t+r) - \arctan (t-r).$$
Choosing the conformal factor
$$\Omega = \frac{2}{\sqrt{1+u^2}\sqrt{1+v^2}} = \frac{2}{\sqrt{1+(t-r)^2}\sqrt{1+(t+r)^2}},$$
we obtain the rescaled metric
\begin{equation}\label{Einstein_metric}
\hat g = \Omega^2 g = \d\tau^2 - \d \zeta^2 - \sin^2\zeta\d\omega^2,
\end{equation}
and the full conformal compactification of Minkowski spacetime is described by the domain
$$\hat{\mathbb M} = \left\{|\tau| + \zeta \leq \pi \; , \; \zeta \geq 0 \; , \; \omega\in S^2  \right\}.$$
We notice that the rescaled metric $\hat g$ can be extended analytically to the whole Einstein's cylinder $\mathfrak{C} = \mathbb R_\tau \times S^3_{\zeta, \theta, \varphi}$. The full conformal boundaries of Minkowski spacetime are described by
\begin{itemize}
\item[$\bullet$] The future and fast null infinities are
$$\scri^+ = \left\{ (\tau, \zeta,\omega); \; \tau + \zeta = \pi, \; \zeta \in ]0,\pi[, \; \omega \in S^2  \right\},$$
$$\scri^- = \left\{ (\tau, \zeta,\omega); \; \tau - \zeta = \pi, \; \zeta \in ]0,\pi[, \; \omega \in S^2  \right\},$$
which are smooth null hypersurfaces for $\hat g$.
\item[$\bullet$] The future and past timelike infinities are
$$i^\pm = \left\{ (\tau = \pm \pi, \; \zeta = 0, \; \omega); \; \omega \in S^2 \right\},$$
which are smooth points for $\hat g$.
\item[$\bullet$] The spacelike infinity is
$$i_0 = \left\{ (\tau = 0, \; \zeta = \pi, \; \omega); \; \omega \in S^2 \right\},$$
which is also a smooth point for $\hat g$.
\end{itemize}
The hypersurface $\Sigma_0=\left\{t = 0 \right\}$ in Minkowski spacetime is described by the $3-$sphere $S^3 = \left\{\tau = 0 \right\}$ excluding the point $i_0$ on Einstein's cylinder, i.e, $S^3 = \left\{t = 0 \right\} \cup i_0$.

We can choose the Newman-Penrose tetrad normalization as follows:
\begin{equation}\label{NMP}
\hat l^a = \frac{1}{\sqrt 2} \left(\partial_\tau + \partial_\zeta \right), \; \hat n^a = \frac{1}{\sqrt 2} \left( \partial_\tau - \partial_\zeta \right), \; \hat m^a = \frac{1}{\sqrt 2\sin\zeta} \left( \partial_\theta + \frac{i}{\sin\theta} \partial_\varphi \right).
\end{equation}
Therefore, we have
$$\partial_t = (1 + \cos\tau \cos\zeta)\partial_\tau - \sin\tau\sin\zeta\partial_\zeta,$$
$$\partial_r = - \sin\tau \sin\zeta \partial_\tau + (1 + \cos\tau \cos\zeta)\partial_\zeta.$$
The vector field $\partial_t$ is normal to the null hypersurface $\scri^\pm$ and tends to zero at $i^\pm$. We have also the following relations
$$\hat l^a = \frac{1+v^2}{2}l^a, \; \hat n^a = \frac{1+u^2}{2} n^a.$$

In the term of the associated spin-frame we have
$$\hat o^A = \sqrt{\frac{1+v^2}{2}} o^A = \Omega_2^{-1} o^A, \; \hat\iota^A = \sqrt{\frac{1+u^2}{2}}\iota^A = \Omega_1^{-1} \iota^A,$$
where
$$\Omega_1 = \sqrt{\frac{2}{1+u^2}}, \; \Omega_2 = \sqrt{\frac{2}{1+v^2}}.$$
Since $\hat l^a \hat n_a = l^a n_a = 1$, we obtain the relation of the dual $\left\{ o_A,\iota_A\right\}$ and its rescaling $ \left\{\hat o_A,\hat\iota_A \right\}$ is
$$\hat o_A = \Omega_1 o_A, \; \hat\iota_A = \Omega_2 \iota_A.$$

The rescaled scalar curvature is
$$\frac{1}{6}\mathrm{Scal}_{\mathfrak{C}}=1.$$

The volume form associated with the rescaled metric ${\hat g}$ is
$$\mathrm{dVol}^4_{{\hat g}} = \sqrt{|\hat{g}|}\d \tau \d \zeta\d^2\omega = \sin^2\zeta\d \tau \d \zeta\d^2\omega = \d \tau \d\mu_{S^3},$$
where $\d \mu_{S^3} = \sin^2\zeta\d\zeta\d^2\omega$ is the volume form of $3-$sphere $S^3$ with the Euclidean metric 
$$\sigma^2_{S^3} = \d \zeta^2 + \sin^2\zeta \d\omega^2.$$

We can calculate the spin coefficients by using the Ricci rotation coefficients (see \cite{Ni1995}) and get
\begin{gather}
\hat\kappa = \hat\varepsilon = \hat\sigma = \hat\gamma = \hat\lambda = \hat\tau = \hat\nu = \hat\pi = 0 , \label{spin-full-1}\\
\hat\rho = \hat\mu = \frac{\cot \zeta}{\sqrt 2}, \; \hat\alpha = - \hat\beta = -\frac{\cot\theta}{2\sqrt 2\sin\zeta}.\label{spin-full-2}
\end{gather}

\subsection{The partial conformal compactification}
The full conformal compactification spacetime $\hat{\mathbb{M}}$ and Einstein's cylinder $\mathfrak{C}$ are useful to prove the wellposedness of the Cauchy problem. However, since the scalar curvature is non zero on $\mathfrak{C}$, it is not convenient to study the wellposedness of the Goursat problem in $\hat{\mathbb{M}}$. Therefore, we consider the partial conformal compactification of Minkowski spacetime that has zero scalar curvature and is useful to prove the wellposedness the Goursat problem in Section \ref{S6}.

Consider the retard time variable $u = t-r$ and the conformal factor $\widetilde{\Omega} = 1/r = R$ we obtain the following expression for the rescaled metric $\widetilde{g}$:
\begin{equation}
\widetilde{g} = R^2g = R^2 \d u^2 - 2 \d u \d R - \d \omega^2.
\end{equation}
Since the rescaled metric $\widetilde{g}$ can be extended as an analytic metric on the domain $\mathbb{R}_u\times [0,+\infty[_R\times S^2_{\theta,\varphi}$, we can add to Minkowski spacetime the boundary $\mathbb{R}_u \times \left\{R=0\right\}\times S^2_{\theta,\varphi}$. As $r$ goes to $+\infty$, a point on this boundary $(u=u_0,R=0,\theta = \theta_0,\varphi=\varphi_0)$ is reached along an outgoing radial null geodesic
$$\gamma_{u_0,\theta_0,\varphi_0}(r) = (t=r+u_0,r,\theta=\theta_0,\varphi=\varphi_0).$$
This boundary describes the future null infinity $\scri^+$: 
$$\scri^+ = \mathbb R_u \times \left\{ R = 0 \right\} \times S_\omega^2.$$

Similarly, we can use an advanced time variable $v=t+r$ and the conformal factor $\widetilde{\Omega} = R$ to construct the rescaled metric and get the past null infinity $\scri^-$ as a boundary of Minkowski spacetime. We have three following points at infinity
\begin{itemize}
\item[$\bullet$] The future (resp. past) timelike infinity point $i^+$ (res. $i^-$) defined as the limit point of uniformly timelike curves as $t$ tend to $+\infty$ (resp. $-\infty$) is
$$i^\pm = \left\{ (u = \pm\infty, R = 0, \omega); \; \omega \in S^2    \right\}.$$
\item[$\bullet$] The spacelike infinity point $i_0$ defined as the limit point of uniformly spacelike curves as $r$ tend to $+\infty$ is
$$i_0 = \left\{ (u = \mp\infty, R = 0, \omega); \; \omega \in S^2 \right\}.$$
\end{itemize}
The partial conformal compactification can be described by the domain
$$\widetilde{\mathbb M} = \mathbb M \cup \scri^\pm.$$
\begin{remark}
We notice that the null infinity hypersurface $\scri^\pm$ in $\widetilde{\mathbb{M}}$ are the same null infinity hypersurfaces in $\hat{\mathbb{M}}$ due to they are reached by the outgoing (resp. incoming) radial null geodesics. The difference between the two conformal structures is that the points $i^\pm$ and $i_0$ are infinite in $\widetilde{\mathbb{M}}$ and finite in $\hat{\mathbb{M}}$.
\end{remark}

We chose the Newman-Penrose tetrad normalization as 
$$\widetilde{l}^a = -\frac{1}{\sqrt 2}\partial_R, \; \widetilde{n}^a = \sqrt 2 \left(\partial_u + \frac{R^2}{2}\partial_R \right), \; \widetilde{m}^a = \frac{1}{\sqrt 2}  \left( \partial_\theta + \frac{i}{\sin\theta}\partial_\varphi \right).$$
This shows that
$$\widetilde{l}^a = r^2 l^a, \; \widetilde{n}^a = n^a, \; \widetilde{m}^a = rm^a.$$
In the term of the associated spin-frame we have
$$\widetilde{o}^A = ro^A, \; \widetilde{\iota}^A = \iota^A, \; \widetilde{o}_A = o_A, \; \widetilde{\iota}_A = R\iota_A.$$

The rescaled scalar curvature is
$$\mathrm{Scal}_{\widetilde{g}}=0.$$
The volume form associated with the rescaled metric $\widetilde{g}$ is
$$\mathrm{dVol}^4_{\widetilde{g}} = \widetilde{\Omega}^4\mathrm{dVol}^4_{g} = R^2\d t \d r\d^2\omega = - \d t \d R\d^2\omega.$$
The rescaled spin coefficients can be calculated as (see \cite{MaNi2012} for the case $M > 0$):
\begin{equation}\label{spin-part-1}
\widetilde{\kappa} = \widetilde{\varepsilon} = \widetilde{\sigma} = \widetilde{\lambda} = \widetilde{\tau} = \widetilde{\nu} = \widetilde{\pi} = \widetilde{\rho} = \widetilde{\mu} = 0,
\end{equation}
\begin{equation}\label{spin-part-2}
\widetilde{\gamma} = \frac{R}{\sqrt 2}, \; \widetilde{\alpha} = - \widetilde{\beta} = -\frac{\cot\theta}{2\sqrt 2}.
\end{equation}

\section{The spin-$n/2$ zero rest-mass fields}\label{S3}
In this section, we give the detailed expressions of the generalized spin-$n/2$ zero rest-mass fields and equations in the term of spin components on the full and partial conformal compactification spacetimes. These expressions will be used to prove the wellposedness of Cauchy and Goursat problems, establish the pointwise decays and calculate the energy of the fields in sections \ref{S4}, \ref{S5} and \ref{S6}.

\subsection{The original equations}
Since the total symmetry of $\phi_{\underbrace{AB...F}_{n \; indexs}} = \phi_{\underbrace{(AB...F)}_{n \; indexs}}$, we have the formula of the spin-$n/2$ zero rest-mass field as follows
\begin{eqnarray}\label{Ofield}
\phi_{AB...F} &=& \phi_n o_Ao_B...o_F - \phi_{n-1}(\iota_Ao_B...o_F + ... + o_Ao_B...\iota_F)\cr
&& + ... + (-1)^{k}\phi_{n-k}\sum_{k=1}^{n-1} \underbrace{\iota_A\iota_B...\iota_C}_{k \; terms} \underbrace{o_D...o_F}_{n-k \; terms} + ... + (-1)^n \phi_0 \iota_A\iota_B...\iota_F,
\end{eqnarray}
where $\phi_k = \phi_{\underbrace{00...0}_{ n-k \; terms} \underbrace{11..1}_{k \; terms}} \; (0\leq k \leq n)$ is a contraction of $\phi_{AB...F}$ with $(n-k)$ omicrons $o^M$ and $k$ iotas $\iota^N$, where $M,\, N \in \left\{A,\,B...,F \right\}$. 

We recall that a scalar function $\eta$ is said to have weight $\left\{ r',\,r;\, t',\, t \right\}$
if under a rescaling of the spin-frame by nowhere vanishing complex scalar fields $\lambda$ and $\mu$ (see \cite[eq. (4.12.9), pp. 253, Vol. 1]{PeRi}):
$$o^A \mapsto \lambda o^A,\, \iota^A \mapsto \mu \iota^A,$$
it transforms as
$$\eta \mapsto \lambda^{r'}\mu^{r}\bar{\lambda}^{t'}\bar{\mu}^t.$$
If we consider normalized spin-frames (this requires that $\lambda\mu=1$ for preserving the normalisation), then we said that $\eta$ has type $\left\{p=r'-r,\, q= t'-t \right\}$ (see \cite[eq. (4.12.10), pp. 253, Vol. 1]{PeRi}).
Therefore, we can see that the weighted function $\phi_k$ has weight $(n-k,k;0,0)$ or simply has type $(p=n-2k,q=0)$.

Using \eqref{Ofield} we have
\begin{equation}
\phi_{AB...F}\bar{\phi}_{A'B'...F'} = |\phi_{n-k}|^2 \sum_{k=0}^n \underbrace{n_an_b...n_c}_{k \; terms} \underbrace{l_d...l_f}_{n-k \; terms} + A,
\end{equation}
where $A$ is the sum of the components involving $m_a$ or $\bar{m}_a$.

Using the Geroch-Held-Penrose formalism, the spin-$n/2$ massless equation $\nabla^{AA'}\phi_{AB...F} = 0$ has the following expression (see \cite[eq. $(4.12.44)$, pp. $260$, Vol. 1]{PeRi}):
\begin{align}\label{dirac}
\begin{cases}
\thorn \phi_k - \eth'\phi_{k-1} &= -(k-1)\lambda\phi_{k-2} + k\pi\phi_{k-1}+(n-k+1)\rho\phi_k - (n-k)\kappa\phi_{k+1},\\
\thorn'  \phi_{k} - \eth\phi_{k+1} &= (n-k-1)\sigma\phi_{k+2} - (n-k)\tau\phi_{k+1} - (k+1)\mu\phi_{k} + k\nu\phi_{k-1}, 
\end{cases}
\end{align}
where $k=1,2...n$ in the first equation and $k=0,1...n-1$ in the second equation and
\begin{align*}
\thorn \phi_k &= (l^a\partial_a - (n-2k)\varepsilon)\phi_k,\cr
\eth'\phi_{k-1} &= ({\bar m}^a\partial_a - (n-2k+2)\alpha)\phi_{k-1},\cr
\thorn'\phi_{k} &= ({n}^a\partial_a - (n-2k)\gamma)\phi_{k},\cr
\eth\phi_{k+1} &= (m^a\partial_a - (n-2k-2)\beta)\phi_{k+1}. 
\end{align*}

\subsection{The rescaled equations}
Under the conformal transformations of Minkowski metric $\hat{g}= \Omega^2g$ and $\widetilde{g}=\widetilde{\Omega}^2g$ we have  (see \cite[eq. (5.7.20), pp. 366, Vol. 1]{PeRi}):
$$\hnabla^{AA'}\hat{\phi}_{AB...F} = \Omega^{-3}\nabla^{AA'}\phi_{AB...F} ,\, \widetilde{\nabla}^{AA'}\widetilde{\phi}_{AB...F}=\widetilde{\Omega}^{-3}\nabla^{AA'}\phi_{AB...F},$$
where $\hat{\phi}_{AB...F} = \Omega^{-1}\phi_{AB...F}$ and $\widetilde{\phi}_{AB...F}=\widetilde{\Omega}^{-1}\phi_{AB...F}$. These equalities show that the equation $\nabla^{AA'}\phi_{AB...F} = 0$ is conformal invariant. In particular, if $\phi_{AB...F}$ satisfies the equation $\nabla^{AA'}\phi_{AB...F} = 0$, then $\hat{\phi}_{AB...F}$ and $\widetilde{\phi}_{AB...F}$ satisfy the following rescaled equations 
$$\hnabla^{AA'}\hat{\phi}_{AB...F} = 0,\, \widetilde{\nabla}^{AA'}\widetilde{\phi}_{AB...F}=0$$
respectively. Therefore, we can apply expression \eqref{dirac} to establish the rescaled equation in the full and partial conformal compactification spacetimes.

In $\hat{\mathbb{M}}$ the rescaled spin-frame $\left\{ \hat o_A,\hat \iota_A  \right\}$ is given by
$$\hat o_A = \Omega_1 o_A, \; \hat\iota_A = \Omega_2 \iota_A.$$
Therefore, we have
\begin{eqnarray*}
\hat\phi_{AB...F} &=& \Omega^{-1}\phi_{AB...F} = \Omega_1^{-1}\Omega_2^{-1} \phi_{AB...F} \cr
&=& \Omega_1^{-1}\Omega_2^{-1}\phi_n o_Ao_B...o_F - \Omega_1^{-1}\Omega_2^{-1}\phi_{n-1}(\iota_Ao_B...o_F + ... + o_Ao_B...\iota_F)\cr
&& + ... + \Omega_1^{-1}\Omega_2^{-1}(-1)^{k}\phi_{n-k}\sum_{k=1}^{n-1} \underbrace{\iota_A\iota_B...\iota_C}_{k \; terms}\underbrace{o_D...o_F}_{n-k \; terms} \cr
&& + ... + \Omega_1^{-1}\Omega_2^{-1} (-1)^n \phi_0 \iota_A\iota_B...\iota_F \cr
&=& \Omega_1^{-1-n}\Omega_2^{-1}\phi_n \hat o_A\hat o_B...\hat o_F - \Omega_1^{-n}\Omega_2^{-2}\phi_{n-1}(\hat\iota_A \hat o_B...\hat o_F + ... + \hat o_A\hat o_B...\hat \iota_F)\cr
&& + ... + \Omega_1^{-1-k}\Omega_2^{-1-(n-k)}(-1)^{k}\phi_{n-k}\sum_k \underbrace{\hat\iota_A\hat\iota_B...\hat\iota_C}_{k \; terms}\underbrace{\hat o_D...\hat o_F}_{n-k \; terms} \cr
&& + ... + \Omega_1^{-1}\Omega_2^{-1-n} (-1)^n \phi_0 \hat\iota_A\hat\iota_B...\hat\iota_F.
\end{eqnarray*}
This leads to
\begin{equation}
\hat\phi_{n-k} = \Omega_1^{-1-k}\Omega_2^{-1-(n-k)}(-1)^{k} \phi_{n-k} \; , \; 0 \leq k \leq n.
\end{equation}

Plugging the spin coefficients \eqref{spin-full-1} and \eqref{spin-full-2} into the expression \eqref{dirac}, we get the rescaled equation $\hnabla^{AA'} \hat{\phi}_{AB...F} = 0$ in $\hat{{\mathbb M}}$ as follows
\begin{align}\label{dirac-full}
\begin{cases}
\frac{1}{\sqrt 2}(\partial_\tau + \partial_\zeta) {\hat\phi}_k - \frac{1}{\sqrt 2\sin\zeta}\left( \partial_\theta - \frac{i}{\sin\theta}\partial_\varphi + (n-2k+2)\frac{\cot\theta}{2}  \right) {\hat\phi}_{k-1} &= (n-k+1)\frac{\cos\zeta}{2} {\hat\phi}_k,\cr
\frac{1}{\sqrt 2}(\partial_\tau - \partial_\zeta) {\hat\phi}_k - \frac{1}{\sqrt 2\sin\zeta}\left( \partial_\theta + \frac{i}{\sin\theta}\partial_\varphi  - (n-2k-2)\frac{\cot\theta}{2}  \right) {\hat\phi}_{k+1} &= -(k+1)\frac{\cos\zeta}{2} {\hat\phi}_k,
\end{cases}
\end{align}
where $k=1,2...n$ in the first equation and $k=0,1...n-1$ in the second equation.

On the other hand, in $\widetilde{\mathbb{M}}$ we have $\widetilde{\Omega} = 1/r$ and the rescaled spin-frame $\left\{ \widetilde{o}_A,\, \widetilde{\iota}_A  \right\}$ is given by
$$\widetilde{o}_A = o_A \; , \; \widetilde{\iota}_A = \widetilde{\Omega}\iota_A = R\iota_A.$$
Therefore, we have
\begin{eqnarray}
\widetilde{\phi}_{AB...F} &=& \widetilde{\Omega}^{-1} \phi_{AB...F} \cr
&=& \widetilde{\Omega}^{-1} \phi_n o_Ao_B...o_F - \widetilde{\Omega}^{-1} \phi_{n-1}(\iota_Ao_B...o_F + ... + o_Ao_B...\iota_F)\cr
&& + ... + \widetilde{\Omega}^{-1}(-1)^{k}\phi_{n-k}\sum_k \underbrace{\iota_A\iota_B...\iota_C}_{k \; terms} \underbrace{o_D...o_F}_{n-k \; terms} \cr
&& + ... + \widetilde{\Omega}^{-1} (-1)^n \phi_0 \iota_A\iota_B...\iota_F \cr
&=& r \phi_n \widetilde{o}_A \widetilde{o}_B...\widetilde{o}_F - r^2  \phi_{n-1}(\widetilde{\iota}_A\widetilde{o}_B...\widetilde{ o}_F + ... + \widetilde{o}_A\widetilde{o}_B...\widetilde{\iota}_F)\cr
&& + ... + r^{1+k}(-1)^{k}\phi_{n-k}\sum_k \underbrace{\widetilde{\iota}_A\widetilde{\iota}_B...\widetilde{\iota}_C}_{k \; terms}\underbrace{\widetilde{o}_D...\widetilde{o}_F}_{n-k \; terms} \cr
&& + ... + r^{1+n} (-1)^n \phi_0 \widetilde{\iota}_A\widetilde{\iota}_B...\widetilde{\iota}_F.
\end{eqnarray}
This leads to
\begin{equation}\label{eq-part}
\widetilde{\phi}_{n-k} = r^{1+k} \phi_{n-k}, \; 0 \leq k \leq n.
\end{equation}

Plugging the spin coefficients \eqref{spin-part-1} and \eqref{spin-part-2} into the expression \eqref{dirac}, we get the rescaled equation $\widetilde{\nabla}^{AA'}\widetilde{\phi}_{AB...F} = 0$ in $\widetilde{\mathbb M}$ as follows
\begin{align}\label{dirac-part}
\begin{cases}
-\frac{1}{\sqrt 2}\partial_R \widetilde{\phi}_k - \frac{1}{\sqrt 2}\left( \partial_\theta - \frac{i}{\sin\theta}\partial_\varphi + (n-2k+2)\frac{\cot\theta}{2}  \right) \widetilde{\phi}_{k-1} &= 0,\cr
\left(\sqrt 2 \partial_u + \frac{R^2}{\sqrt 2}\partial_R - (n-2k)\frac{R}{\sqrt 2} \right) \widetilde{\phi}_k - \frac{1}{\sqrt 2}\left( \partial_\theta + \frac{i}{\sin\theta}\partial_\varphi  - (n-2k-2)\frac{\cot\theta}{2}  \right) \widetilde{\phi}_{k+1} &= 0, 
\end{cases}
\end{align}
where $k=1,2...n$ in the first equation and $k=0,1...(n-1)$ in the second equation.

\section{The Cauchy problem and decays of the fields}\label{S4}
In this section we prove the wellposedness of the Cauchy problem for the rescaled equation $\hnabla^{AA'}\hat{\phi}_{AB...F} = 0$ in the whole Einstein's cylinder $\mathfrak{C} = \mathbb{R}\times S^3_{\zeta,\theta,\varphi}$. As consequences, we obtain the wellposedness of the rescaled equations $\hnabla^{AA'}\hat{\phi}_{AB...F} = 0$ and $\widetilde{\nabla}^{AA'}\widetilde{\phi}_{AB...F}=0$ in $\hat{\mathbb M}$ and $\widetilde{\mathbb M}$ respectively. We study also the pointwise decay, i.e, the decays in time of the components of the rescaled solution $\widetilde{\phi}_{AB...F}$ that is useful to prove the energy equality in $\widetilde{\mathbb{M}}$ in Section \ref{S5}.

\subsection{Wellposedness of Cauchy problem}
The Cauchy problem of the rescaled massless equation with the initial data on $S^3 = \left\{ \tau = 0 \right \}$ in $\mathfrak{C}$ reads
\begin{align}\label{Cauchy}
\begin{cases}
{\hnabla}^{AA'} {\hat\phi}_{AB...F} &= 0,\cr
{\hat\phi}_{AB...F}|_{S^3} &= {\hat \psi}_{AB...F} \in \mathcal{C}^\infty(S^3,\mathbb{S}_{(AB...F)}) \cap \mathcal{D},
\end{cases}
\end{align}
where $\mathcal{D}$ is the constraint space on $S^3 = \left\{ \tau = 0 \right \}$:
$$ \mathcal{D} = \left\{ \hat{\psi}_{AB...F} \in L^2(S^3,\mathbb{S}_{(AB...F)}): \,  D^{AB} {\hat\psi}_{ABC...F} = 0  \right\},$$
where $D^{AB}$ is the intrinsic space spinor derivative on $S^3 = \left\{ \tau = 0 \right \}$ (see Appendix \ref{app_3_existe_constraint}). The spin-$1/2$ case, i.e, Dirac equation $\hnabla^{AA'}\hat{\phi}_A=0$ does not have constraints.

We notice that $D^{AB}\left( {\hat\phi}_{ABC...F}|_{\tau = constant}\right)$ can be also understood as the projection of ${\hnabla}^{AA'} {\hat\phi}_{AB...F} = 0$ on the future-oriented timelike vector $\mathcal{T}^a = \sqrt{2}\partial_\tau$ (see also \cite[Section 2.1]{ABJ}, \cite[Section 2.2 and Remark 2.3]{MaNi2004}), i.e,
$$D^{AB} \left( {\hat\phi}_{ABC...F}|_{\tau = constant}\right) = \left(\mathcal{T}^{AA'} \hnabla_{A'}^B {\hat\phi}_{BAC...F}\right)|_{\tau=constant}.$$
Since the wellposedness of the spin-$1/2$ equation was established in \cite{Ni2002}, we consider the wellposedness of spin-$n/2$ equations with $n\geq 2$. 
In these cases the existence of non-trivial solutions of the constraint system is given in Appendix \ref{app_3_existe_constraint}.

We state and prove the wellposedness of the generalized spin-$n/2$ zero rest-mass equation $(n\geq 2)$ in the following theorem:
\begin{theorem}\label{cauchyproblem}(Cauchy problem)
The Cauchy problem for the rescaled massless
 equation \eqref{Cauchy} in $\mathfrak{C}$ is well-posed, i.e, for any ${\hat\psi}_{AB...F} \in \mathcal{C}^\infty(S^3, \mathbb{S}_{(AB...F)}) \cap \mathcal{D}$ there exists a unique ${\hat\phi}_{AB...F}$ solution of ${\hnabla}^{AA'}{\hat\phi}_{AB...F} = 0$ such that   
$${\hat\phi}_{AB...F} \in {\mathcal C}^\infty(\mathfrak{C}, \mathbb{S}_{(AB...F)}) \; ; \; {\hat\phi}_{AB...F}|_{\tau=0} = {\hat\psi}_{AB...F}.$$
\end{theorem}
\begin{proof}
First, we show that the equation ${\hnabla}^{AA'} {\hat\phi}_{AB...F} =0$ can split into the constraint equation
\begin{equation}\label{con}
\left(\mathcal{T}^{AA'} {\hnabla}_{A'}^B \hat{\phi}_{ABC...F} \right)|_{\tau=constant} =0.
\end{equation}
for all $\tau$ and a symmetric hyperbolic evolution system. 

Indeed, the equation ${\hnabla}^{AA'} {\hat\phi}_{AB...F} =0$ can be expressed as a set of $2n$ scalar equations on the spin components of $\hat{\phi}_{AB...F}$ (see equation \eqref{dirac-full}):
\begin{align}\label{D}
\begin{cases}
\frac{1}{\sqrt 2}(\partial_\tau + \partial_\zeta) {\hat\phi}_k - \frac{1}{\sqrt 2\sin\zeta}\left( \partial_\theta - \frac{i}{\sin\theta}\partial_\varphi + (n-2k+2)\frac{\cot\theta}{2}  \right) {\hat\phi}_{k-1} &= (n-k+1)\frac{\cos\zeta}{2} {\hat\phi}_k \cr
&\mbox{with} \; 1 \leq k \leq n,\cr
\frac{1}{\sqrt 2}(\partial_\tau - \partial_\zeta) {\hat\phi}_k - \frac{1}{\sqrt 2\sin\zeta}\left( \partial_\theta + \frac{i}{\sin\theta}\partial_\varphi  - (n-2k-2)\frac{\cot\theta}{2}  \right) {\hat\phi}_{k+1} &= -(k+1)\frac{\cos\zeta}{2} {\hat\phi}_k \cr
&\mbox{with} \; 0 \leq k \leq n-1.
\end{cases}
\end{align}

The constrain system in the term of spin components is obtained by taking the differences of $(n-1)$ couples of the equations of the above system corresponding to $k=1,2...n-1$ respectively. Therefore, we have the constraint system which consists $(n-1)$ equations as follows
\begin{gather}\label{CONE}
{\sqrt 2}\partial_\zeta {\hat\phi}_k - \frac{1}{\sqrt 2\sin\zeta}\left( \partial_\theta - \frac{i}{\sin\theta}\partial_\varphi + (n-2k+2)\frac{\cot\theta}{2}  \right) {\hat\phi}_{k-1} \cr 
+ \frac{1}{\sqrt 2\sin\zeta}\left( \partial_\theta + \frac{i}{\sin\theta}\partial_\varphi  - (n-2k-2)\frac{\cot\theta}{2}  \right) {\hat\phi}_{k+1} = (n+2)\frac{\cos\zeta}{2} {\hat\phi}_k,
\end{gather}
for $k=1,2...n-1$. In the term of spinor derivatives this system corresponds to \eqref{con} (see Appendix \ref{app_3_existe_constraint}).

In order to obtain the evolution system we keep the first and the last equations of \eqref{D} that correspond to $k=n$ and $k=0$ respectively. Then we obtain $(n-1)$ equations which are the sums of two equations in $(n-1)$ couples of the equations of \eqref{D} corresponding to $k=1,2...n-1$ respectively. Therefore, we get the evolution system which consists $(n+1)$ equations as follows
\begin{align}\label{evolution}
\begin{cases}
\frac{1}{\sqrt 2}(\partial_\tau + \partial_\zeta) {\hat\phi}_n - \frac{1}{\sqrt 2\sin\zeta}\left( \partial_\theta - \frac{i}{\sin\theta}\partial_\varphi - (n-2)\frac{\cot\theta}{2}  \right) {\hat\phi}_{n-1} &= \frac{\cos\zeta}{2} {\hat\phi}_n,\cr
{\sqrt 2}\partial_\tau {\hat\phi}_k - \frac{1}{\sqrt 2\sin\zeta}\left( \partial_\theta - \frac{i}{\sin\theta}\partial_\varphi + (n-2k+2)\frac{\cot\theta}{2}  \right) {\hat\phi}_{k-1}  \cr
- \frac{1}{\sqrt 2\sin\zeta}\left( \partial_\theta + \frac{i}{\sin\theta}\partial_\varphi  - (n-2k-2)\frac{\cot\theta}{2}  \right) {\hat\phi}_{k+1} &= (n-2k)\frac{\cos\zeta}{2} {\hat\phi}_k \cr
\frac{1}{\sqrt 2}(\partial_\tau - \partial_\zeta) {\hat\phi}_0 - \frac{1}{\sqrt 2\sin\zeta}\left( \partial_\theta + \frac{i}{\sin\theta}\partial_\varphi  - (n-2)\frac{\cot\theta}{2}  \right) {\hat\phi}_{1} &= -\frac{\cos\zeta}{2} {\hat\phi}_0,
\end{cases}
\end{align} 
where $k=1,2...n-1$.

We rewrite the evolution system above under the matrix form by putting
$$\Phi = \left( \begin{matrix}\hat{\phi}_n\\ \hat{\phi}_{n-1}\\...\\ \hat{\phi}_0 \end{matrix} \right).$$
The effect of the derivative operator on $\Phi$, can be understood as the effect on each components of $\Phi$, for instance
$$\partial_\tau\Phi = \left( \begin{matrix}\partial_\tau\hat{\phi}_n\\ \partial_\tau\hat{\phi}_{n-1}\\...\\ \partial_\tau\hat{\phi}_0 \end{matrix} \right).$$
The matrix coefficient with $\partial_\tau$ and $\partial_\zeta$ are $(n+1)\times (n+1)-$matrix diagrams
$$A= \left(\begin{matrix}
\frac{1}{\sqrt 2}&&0&&...&&0&&0\\
0&&\sqrt 2&&...&&0&&0\\
...&&...&&...&&...&&...\\
0&&0&&...&&\sqrt 2&&0\\
0&&0&&...&&0&&\frac{1}{\sqrt 2}
\end{matrix} \right) \; , \; B= \left(\begin{matrix}
 \frac{1}{\sqrt 2}&&0&&...&&0&&0\\
0&&0&&...&&0&&0\\
...&&...&&...&&...&&...\\
0&&0&&...&&0&&0\\
0&&0&&...&&0&&-\frac{1}{\sqrt 2}
\end{matrix} \right)$$
respectively. The matrix coefficients associated with $\partial_\theta$ and $\partial_\varphi$ are
$$C= \left(\begin{matrix}
 0&&-\frac{1}{\sqrt 2\sin\zeta}&&...&&0&&0\\
-\frac{1}{\sqrt 2\sin\zeta}&&0&&...&&0&&0\\
...&&...&&...&&...&&...\\
0&&0&&...&&0&&-\frac{1}{\sqrt 2\sin\zeta}\\
0&&0&&...&&-\frac{1}{\sqrt 2\sin\zeta}&&0
\end{matrix} \right)$$
and
$$D= \left(\begin{matrix}
0&&\frac{i}{\sqrt 2\sin\zeta\sin\theta}&&...&&0&&0\\
-\frac{i}{\sqrt 2\sin\zeta\sin\theta}&&0&&...&&0&&0\\
...&&...&&...&&...&&...\\
0&&0&&...&&0&&\frac{i}{\sqrt 2\sin\zeta\sin\theta}\\
0&&0&&...&&-\frac{i}{\sqrt 2\sin\zeta\sin\theta}&&0
\end{matrix} \right)$$
respectively. 

Therefore, we obtain the matrix form of the evolution system as
$$A\partial_\tau\Phi + B\partial_\zeta\Phi + C\partial_\theta\Phi + D\partial_\varphi\Phi + H\Phi = 0$$
which is equivalent to
\begin{equation}\label{matrix_evolution}
\partial_\tau\Phi + A^{-1}B\partial_\zeta\Phi + A^{-1}C\partial_\theta\Phi + A^{-1}D\partial_\varphi\Phi + A^{-1}H\Phi = 0,
\end{equation}
where the matrix coefficients $A,B,C,D$ are given as above and $H$ is the matrix of zero order terms. The coefficients of $C$ and $D$ are singular at $\xi = 0,\pi$ and $\theta = 0,\pi$. These are coordinates singularities due to the choice of spherical coordinates on $S^3$. The spherical symmetry entails that these are not authentic singularities. Since $A,B,C,D$ are Hermitian and $A$ is diagonal we can easy check that $A^{-1}B, A^{-1}C, A^{-1}D$ are also Hermitian. Therefore, the evolution system \eqref{evolution} is a symmetric hyperbolic system. By using Leray's theorem (see \cite{Le1953}), for the initial data $\hat\psi_{AB...F} \in \mathcal{C}^\infty(S^3,\mathbb{S}_{(AB...F)})$ the evolution system \eqref{evolution} has a unique solution $\hat\phi_{AB...F} \in \mathcal{C}^\infty(\mathfrak{C},\mathbb{S}_{(AB...F)}) $. 

We will show that the solution of evolution system \eqref{evolution} is the solution of the original system \eqref{Cauchy} by checking that the constraints system is conserved under the evolution solution, i.e, $\left(\mathcal{T}^{AA'} {\hnabla}_{A'}^B \hat{\phi}_{ABC...F} \right)|_{\tau=constant} = 0$ for all $\tau$. Indeed, we put
$$ \hnabla_{A'}^B {\hat\phi}_{ABC...F} = {\hat\Xi}_{AA'C...F}.$$
The constraint of $\hnabla^{AA'}{\hat\phi}_{AB...F}$ on $\left\{ \tau = \mbox{constant} \right\}$ is the projection of $\hnabla^{AA'}{\hat\phi}_{AB...F}$ on $\mathcal{T}^a$:
$$\Gamma_{C...F} = \mathcal{T}^a{\hat\Xi}_{aC...F}, \hbox{   where  } a =AA'.$$
Since $\hat\phi_{AB...F}$ is a solution of the evolution system, we have
\begin{eqnarray*}
0= \frac{1}{2}\mathcal{T}^B_{A'} \hnabla_T \hat{\phi}_{ABC...F} &=& \hnabla_{A'}^B {\hat\phi}_{ABC...F} - (\mathcal{T}^{AA'}\hnabla^B_{A'}\hat{\phi}_{ABC...F})\mathcal{T}_{AA'}\cr
&=& {\hat\Xi}_{AA'C...F} - (\mathcal{T}^a{\hat\Xi}_{AA'C...F})\mathcal{T}_a.
\end{eqnarray*}
This leads to
$$\hnabla^{AA'}{\hat\Xi}_{AA'C...F} = \hnabla^{AA'} \left( (\mathcal{T}^a{\hat\Xi}_{AA'C...F})\mathcal{T}_a \right).$$
We have (see \cite[eq. (5.8.1), pp. 366, Vol. 1]{PeRi} or see Equation \eqref{app_3_wave2} in Appendix \ref{app_3_Commutator})
\begin{eqnarray*}
\hnabla^{AA'}{\hat\Xi}_{AA'C...F} &=& \hnabla^{AA'}\hnabla_{A'}^B {\hat\phi}_{ABC...F} = \hnabla^{A'(A}\hnabla_{A'}^{B)} {\hat\phi}_{ABC...F} \cr
&=& -(n-1)\hat\phi_{ABM(C...K}\hat\Psi_{F)}{^{ABM}} = 0,
\end{eqnarray*}
where $\hat{\Psi}_{ABCD}$ is the Weyl conformal spinor of Einstein's metric \eqref{Einstein_metric}, it will be disappeared due to $\hat{\Psi}_{ABCD} = \Psi_{ABCD}$ ($\Psi_{ABCD}$ is invariant under the conformal operator) and $\Psi_{ABCD} = 0$ in Minkowski spacetime. Therefore, we have
\begin{eqnarray}\label{constraints}
0 &=& \hnabla^{AA'} \left( (\mathcal{T}^a{\hat\Xi}_{aC...F})\mathcal{T}_a \right) \cr
& =& \mathcal{T}_a \hnabla^a(\mathcal{T}^a{\hat\Xi}_{aC...F}) + (\mathcal{T}^a{\hat\Xi}_{aC...F}) \hnabla^a \mathcal{T}_a \cr
& =& \partial_\tau (\mathcal{T}^a{\hat\Xi}_{aC...F}) + (\mathcal{T}^a{\hat\Xi}_{aC...F})\hnabla^a \mathcal{T}_a \cr
& =& \frac{1}{\sqrt 2} (\hat D + \hat {D}')(\mathcal{T}^a{\hat\Xi}_{aC...F}) + (\mathcal{T}^a{\hat\Xi}_{aC...F}) \hnabla^a \mathcal{T}_a\cr
& =& \frac{1}{\sqrt 2} (\hat D + \hat {D}')\Gamma_{C...F} + (\mathcal{T}^a{\hat\Xi}_{aC...F}) \hnabla^a \mathcal{T}_a.
\end{eqnarray}
We have (see \cite[eq. (4.5.26), pp. 227, Vol. 1]{PeRi}): 
$$\hat D\hat{o}_A = \hat\varepsilon \hat{o}_A - \hat\kappa \hat{\iota}_A = 0.$$
By the same way
$$\hat D\hat{\iota}_A = -\hat{\varepsilon}\hat{\iota}_A + \hat{\pi}\hat{o}_A = 0,$$
$$\hat{D}' \hat{o}_A = \hat{\gamma}\hat{o}_A - \hat{\tau}\hat{\iota}_A = 0,$$
$$\hat{D}' \hat{\iota}_A = -\hat{\gamma}\hat{o}_A + \hat{\nu}\hat{\iota}_A = 0.$$
Therefore, $\hat{D} + \hat{D}'$ acts only on the weighted scalar coefficients of the spinor field $(\mathcal{T}^a{\hat\Xi}_{aC...F})$. Using the fact that  $v=\hnabla^a \mathcal{T}_a = 0$ (due to $\partial_\tau$ is a Killing vector field), by projecting the equation \eqref{constraints} on the spin-frame $\left\{ \hat{o}_A,\hat{\iota}_A \right\}$, we get its scalar form as follows
$$\frac{1}{\sqrt 2} \partial_\tau \hat\Gamma = 0$$
where $\hat\Gamma$ is the matrix components of $\mathcal{T}^a{\hat\Xi}_{aC...F}$. This equation has a unique solution $\hat\Gamma=constant$. Since the initial condition $\hat\Gamma|_{\tau = 0} = \left(\mathcal{T}^{AA'} {\hnabla}_{A'}^B \hat{\phi}_{BA...F} \right)|_{\tau=0} = 0$, the solution $\hat\Gamma$ is equal to zero, i.e,
$$\left( \mathcal{T}^{AA'}\hnabla_{A'}^B {\hat\phi}_{ABC...F} \right)|_{\tau=constant} = 0 \,\, \mbox{for all} \, \tau.$$
This shows that the constraints system is conserved and our proof is completed.
\end{proof}
Immediately the Cauchy problem of the spin-$n/2$ massless equation is well-posed in the full conformal compactification spacetime $\hat{ \mathbb M}$: 
\begin{cor}
The solution of the system \eqref{Cauchy} in the full conformal compactification $\hat{\mathbb{M}}$ is the constraint of the solution of the system \eqref{Cauchy} in $\mathfrak{C}$ on $\hat{\mathbb{M}}$.
\end{cor}

Considering the Cauchy problem in the partial conformal compactification spacetime $\widetilde{\mathbb {M}}$. Since $i_0$ is still at infinity, we need to suppose that the support of the initial data is compact.
\begin{cor}
The Cauchy problem of the system \eqref{Cauchy} in $\widetilde{\mathbb M}$ with the initial data $\widetilde{\psi}_{AB...F} \in \mathcal{C}^\infty_0(\Sigma_0,\mathbb{S}_{(AB...F)}) \cap \mathcal{D}$ is well-posed, i.e, for any $\widetilde{\psi}_{AB...F} \in \mathcal{C}_0^\infty(\Sigma_0,\mathbb{S}_{(AB...F)}) \cap \mathcal{D}$ there exists a unique $\widetilde{\phi}_{AB...F}$ solution of $\widetilde{\nabla}^{AA'}\widetilde{\phi}_{AB...F} = 0$ such that   
$$\widetilde{\phi}_{AB...F} \in \mathcal{C}^\infty(\widetilde{\mathbb{M}},\mathbb{S}_{(AB...F)}) \; ; \; \widetilde{\phi}_{AB...F}|_{t=0} = \widetilde{\psi}_{AB...F},$$
where we also denote by $\mathcal{D}$ the constraint space on $\Sigma_0 = \left\{ t = 0 \right\}$ in $\widetilde{\mathbb{M}}$.
\end{cor}
\begin{proof}
Using the full conformal mapping, we can transform the domain $\widetilde{\mathbb M}$ into Einstein's cylinder $\mathfrak{C}$. Now the initial data $\hat{\psi}_{AB...F} = \Omega^{-1}\widetilde{\Omega}\widetilde{\psi}_{AB...F}$ is zero in the neighbourhood of $i_0$ which is a smooth point on the cylinder, then we extend the initial data which is zero in the rest of the support. Applying Theorem \ref{cauchyproblem}, we obtain that the solution will be the restriction of that of the Cauchy problem in $\mathfrak{C}$ on $\widetilde{\mathbb M}$.
\end{proof}

\subsection{Pointwise decays}
The decays along the outgoing null geodesics, i.e, the ''peeling-off'' property of the spin-$n/2$ zero rest mass fields in Minkowski spacetime were obtained in \cite{Sa61,Pe1965,Shu}. The pointwise decays, i.e, decays in time of these fields and their derivations were established in \cite{ABJ} via analyzing Hertz potentials. Here, we give another approach to obtain the pointwise decay of spin-$n/2$ zero rest-mass fields by using the full conformal compactification spacetime. The timelike decays obtained below are sufficient to prove that energy equalities of the rescaled field $\widetilde{\phi}_{AB...F}$ between the null conformal boundaries $\scri^\pm$ and the hypersurface $\Sigma_0=\left\{ t=0 \right\}$ in $\widetilde{\mathbb{M}}$. The main theorem of this section is
\begin{theorem}\label{decay}
There exists two constants $C^\pm_k$ such that 
\begin{equation*}
\lim_{t \rightarrow \pm \infty} t^{n+2}\phi_{n-k} = C^\pm_k.
\end{equation*}
In other words, all of the components of spin-$n/2$ zero rest-mass field $\phi_{AB...F}$ decays as $1/t^{n+2}$ along the integral line of $\partial_t$. As a direct consequence of this decay result, on the partial conformal compactification spacetime $\widetilde{\mathbb{M}}$, we have
\begin{equation*}
\lim_{t \rightarrow \pm \infty} \frac{t^{n+2}}{r^{k+1}} \widetilde{\phi}_{n-k} = C^\pm_k.
\end{equation*}
\end{theorem}
\begin{proof}
In the full conformal compactification $\widehat{\mathbb{M}}$, we have $\hat \phi_{AB...F}(i^+) = \lim_{t\rightarrow +\infty}\Omega^{-1} \phi_{AB...F}$. This leads to exist constants $C^+_k \, (0\leq k \leq n)$ satisfying
\begin{eqnarray*}
C^+_k &=& \lim_{t\rightarrow +\infty}\Omega_1^{-(n-k)} \Omega_2^{-k} \Omega^{-1} \phi_{n-k} \cr 
&=& \lim_{t\rightarrow +\infty} \left(\frac{\sqrt{1+(t-r)^2}}{\sqrt 2}\right)^{n-k} \left( \frac{\sqrt{1+(t+r)^2}}{\sqrt 2} \right)^{k} \frac{\sqrt{1+(t-r)^2}\sqrt{1+(t+r)^2}} {2} \phi_{n-k} \cr
&=& \lim_{t\rightarrow +\infty} t^{n+2}\phi_{n-k}.
\end{eqnarray*} 
Similarly, we can show that there exists constants $C^-_k$ such that
$$\lim_{t \rightarrow -\infty} t^{n+2}\phi_{n-k} = C^-_k.$$
The last equations in Proposition \ref{decay} are a direct consequence of the above decay results and the equations \eqref{eq-part}
$$\widetilde{\phi}_{n-k} = r^{1+k} \phi_{n-k}, \; 0 \leq k \leq n.$$
\end{proof}

\section{Energy fluxes}\label{S5}
In this section, we give the detailed calculations for the energy fluxes of the rescaled fields $\hat{\phi}_{AB...F}$ and $\widetilde{\phi}_{AB...F}$ on the full and partial conformal compactification spacetimes $\hat{\mathbb{M}}$ and $\widetilde{\mathbb{M}}$ respectively.
Then, we prove the energy equalities of the rescaled fields between the conformal boundaries $\scri^\pm$ and the hypersurface $S^3 = \left\{ \tau=0 \right\}$ (resp. $\Sigma_0=\left\{ t=0 \right\}$) in $\hat{\mathbb{M}}$ (resp. $\widetilde{\mathbb{M}}$). The energy equalities play an important role to establish the wellposedness of the Goursat problem.

Let ${\mathcal S}$ be a spacelike hypersurface with the future-oriented unit normal vector field $\nu^a$ in Minkowski spacetime $\mathbb M$ we define the current conserved energy by
$$J_a= \phi_{AB...F} \bar{\phi}_{A'B'...F'} \tau^b\tau^c...\tau^f = \left( |\phi_{n-k}|^2 \sum_{k=0}^n \underbrace{n_an_b...n_c}_{k \; terms}\underbrace{l_d...l_f}_{n-k \; terms} + A \right) \tau^b\tau^c...\tau^f,$$
where $\tau^.$ are timelike vector fields, which doesn't change when we changing the metric by using the conformal mapping. 

We have
\begin{align}\label{phi}
{\phi}_{AB...F}\bar{\phi}_{A'B'...F'} &= |{\phi}_n|^2 \underbrace{l_a...l_f}_{n \; terms} + |\phi_0|^2 \underbrace{n_a...n_f}_{n \; terms} \cr
& + \sum_{k=1}^{n-1}|\phi_{n-k}|^2 \left( n_a \underbrace{ n_b...n_c}_{k-1 \; terms} \underbrace{l_d... l_f}_{n-k \; terms}  + l_a \underbrace{n_b...n_c n_d}_{k \; terms} \underbrace{l_e...l_f}_{n-k-1 \; terms}  \right) + A,
\end{align}
where $A$ is the sum of the components that contain $m^a$ or $\bar{m}^a$. We notice that the sum $A$ will be vanished in the energy fluxes due to the normalization condition of Newman-Penrose tetrad.

Since $\mathcal{S}$ is a spacelike hypersurface, we can choose the transversal vector to $\mathcal{S}$ is also $\nu^a$. The energy flux of the spin field $\phi_{AB...F}$ through ${\mathcal S}$ is defined by
\begin{equation}\label{formula-energy}
{\mathcal E}_{\mathcal S} (\phi_{AB...F}) = \int_{\mathcal S} J_a \nu^a (\nu^a \hook \mathrm{dVol}^4_g) =  \int_{\mathcal S} J_a \nu^a \d \mu_{\mathcal{S}}. 
\end{equation}
Here, we denote the contraction of a vector field with the volume form by the symbol $\hook$\footnote{Let $\mathrm{dVol}^n_g = \d x^1\wedge \d x^2...\wedge \d x^n$ be the volume form associated with the metric $g$ in the local coordinates on the manifold $(M,g)$ of dimension $n$. The contraction of a vector field $X^a\partial_a= \sum_{i=1}^nx^i \partial_{x^i}$ with $\mathrm{dVol}^n_g$ is $X^a\hook \mathrm{dVol}^n_g = \sum_{i=1}^n (-1)^{i-1} x^i \d x^1\wedge ... \widehat{\d x^i} \wedge ... \wedge \d x^n.$.}

In the rescaled spacetime $\widehat{\mathbb{M}}$ (or $\widetilde{\mathbb{M}}$) we recall that
$$\hat g := \Omega^2 g, \; \hat \phi_{AB...F}: = \Omega^{-1}\phi_{AB...F}.$$
The current conserved energy is given by
$$\hat J_a= \hat\phi_{AB...F} \bar{\hat\phi}_{A'B'...F'} \tau^b\tau^c...\tau^f$$
and the unit normal vector to ${\mathcal S}$ for $\hat g$ is now
$$\hat\nu^a = \Omega^{-1}\nu^a.$$
We denote by $ \mu_{\mathcal{S}}$ (resp. $\hat {\mu}_{\mathcal{S}}$) the measure induced on $\mathcal{S}$ by $g$ (resp. $\hat g$), then
$$\hat{\mu}_{\mathcal{S}} = \Omega^3{\mu}_{\mathcal{S}} \; .$$
The energy of the rescaled field on ${\mathcal S}$ is 
\begin{equation}\label{e1}
\hat{\mathcal E}_{\mathcal S} (\hat\phi_{AB...F}) = \int_{\mathcal S} \hat J_a \hat\nu^a \d \hat{\mu}_{\mathcal{S}} =  \int_{\mathcal S}  \Omega^{-2}J_a \Omega^{-1}\nu^a \Omega^3\d \mu_{\mathcal{S}} =  {\mathcal E}_{\mathcal S} (\phi_{AB...F}).
\end{equation}
Therefore, if the vector fields $\tau^b,\,\tau^c,...,\tau^f$ don't change, then the energy on a spacelike hypersurface is conformally invariant.

Moreover, we follow the convention used by Penrose and Rindler \cite{PeRi} about the Hodge dual of a 1-form $\alpha$ on a spacetime $({\cal M},g)$ (i.e. a $4-$dimensional Lorentzian manifold that is oriented and time-oriented)
\begin{equation*}
(*\alpha)_{abcd} = e_{abcd}{\alpha}^d,
\end{equation*}
where $e_{abcd}$ is the volume form on $({\cal M},g)$, denoted simply by $\mathrm{dVol}_g$. We shall use the following differential operator of the Hodge star
\begin{equation*}
\d *\alpha = -\frac{1}{4}(\nabla_a\alpha^a)\mathrm{dVol}^4_g.
\end{equation*}
If $\mathcal{S}$ is the boundary of a bounded open set $\Omega$ and has outgoing
orientation, using Stokes theorem, we have
\begin{equation}\label{Stokesformula}
-4\int_{\mathcal{S}}*\alpha = \int_{\Omega}(\nabla_a\alpha^a)\mathrm{dVol}^4_g.
\end{equation}
Let $\hat{\phi}_{AB...F}$ be a solution of $\hat{\nabla}^{AA'}\hat{\phi}_{AB...F} =0$ with smooth and compactly supported initial data on the rescaled spacetime $(\hat{\mathbb{M}},\hat{g})$ (resp. $(\widetilde{\mathbb{M}},\widetilde{g})$). By using \eqref{Stokesformula} we define the rescaled energy flux associated with $\tau^a=\tau^b=...=\tau^f = \partial_t$, across an oriented hypersurface $\mathcal{S}$ as follows
\begin{equation}\label{e}
\hat{\mathcal{E}}_{\mathcal {S}}(\hat{\phi}_{AB...F}) = -4\int_{\mathcal{S}} *\hat{J}_a\d x^a = \int_{\mathcal{S}} \hat{J}_a\hat{N}^a\hat{L}\hook \mathrm{dVol}^4_{\hat g},
\end{equation}
where $\hat{L}$ is a transverse vector to $\mathcal{S}$ and $\hat{N}$ is the normal vector field to $\mathcal{S}$ such that
$\hat{L}^a\hat{N}_a=1$.
\begin{remark}
The formula \eqref{e1} will be useful to calculate the energy fluxes through the Cauchy hypersurface $\Sigma_0=\left\{ t=0\right\}$ in $\widetilde{\mathbb{M}}$. The formula \eqref{e} will be useful to calculate the energy fluxes across the conformal null boundary $\scri^+$ and the hyperboloid hypersurface $\mathcal{S}_T$ in the next sections.
\end{remark}

\subsection{Energy fluxes in the full conformal compactification}
We choose the vectors $\tau^b,\tau^c...\tau^f$ as follows
$$\tau^b = \tau^c = ... = \tau^f = \partial_\tau = \frac{1}{\sqrt 2} (\hat l^a + \hat n^a) \, .$$
The unit normal vector to the hypersurface $\Sigma_0$ is 
$${\hat{\nu}}_{\Sigma_0}^a = \partial_\tau = \frac{1}{\sqrt 2} (\hat l^a + \hat n^a) \, .$$
Combining with the expression \eqref{phi}, we can calculate
\begin{eqnarray*}
\hat{J}_a {\hat{\nu}}_{\Sigma_0}^a &=& \left( \frac{1}{\sqrt 2} \right)^n \left( |\hat{\phi}_{n-k}|^2 \sum_{k=0}^n \underbrace{\hat{n}_a\hat{n}_b...\hat{n}_c}_{k \; terms} \underbrace{\hat{l}_d...\hat{l}_f }_{n-k \; terms}\right) (\hat l^b + \hat n^b) (\hat l^c + \hat n^c)...(\hat l^f + \hat n^f) (\hat l^a + \hat n^a) \cr
&=& \left( \frac{1}{\sqrt 2} \right)^n \sum_{k=0}^{n}C^k_n |\hat\phi_{n-k}|^2.
\end{eqnarray*}
By using \eqref{e1} we have
\begin{equation}\label{energy-sigma-1}
\hat{\mathcal E}_{\Sigma_0}(\hat\phi_{AB...F}) = \left( \frac{1}{\sqrt 2} \right)^n \int_{\Sigma_0} \sum_{k=0}^{n}C^k_n |\hat\phi_{n-k}|^2 \d\mu_{S^3}.
\end{equation}

Since the Cauchy problem of the generalized spin-$n/2$ zero rest-mass equation is well-posed, we can define trace operator $\hat{\phi}_{AB...F}|_{\scri^+}$ on the null infinity hypersurface $\scri^+$ and then determine the energy through this hypersurface. In particular, the normal vector to the future null infinity $\scri^+$ is
$${\hat{\mathcal{N}}}_{\scri^+}^a = \hat n^a = \frac{1}{\sqrt{2}} (\partial_\tau - \partial_\zeta),$$
hence the transversal vector to the null infinity $\scri^+$ is
$${\hat{\mathcal{L}}}_{\scri^+}^a = \sqrt {2} \partial_\tau \, .$$
Using again the expression \eqref{phi}, we have
\begin{eqnarray*}
\hat{J}_a {\hat{\mathcal{N}}}_{\scri^+}^a &=& \left( \frac{1}{\sqrt 2} \right)^{n-1}\left( |\hat{\phi}_{n-k}|^2 \sum_{k=0}^n \underbrace{\hat{n}_a\hat{n}_b...\hat{n}_c}_{k \; terms}\underbrace{\hat{l}_d...\hat{l}_f }_{n-k \; terms}\right) (\hat l^b + \hat n^b) (\hat l^c + \hat n^c)...(\hat l^f + \hat n^f) \hat n^a \cr
&=& \left( \frac{1}{\sqrt 2} \right)^{n-1}\sum_{k=0}^{n-1} C^{k}_{n-1} |\hat\phi_{n-k}|^2.
\end{eqnarray*}
Therefore, use \eqref{e} we have
\begin{eqnarray*}
\hat{\mathcal E}_{\scri^+}(\hat\phi_{AB...F}) &=& \int_{\scri^+} \hat{J}_a {\hat{\mathcal{N}}}_{\scri^+}^a ({\hat{\mathcal{L}}}_{\scri^+}^a \hook \mathrm{dVol^4_{\hat g}}) \cr
&=& \left( \frac{1}{\sqrt 2} \right)^{n-2} \int_{\scri^+}\sum_{k=0}^{n-1} C^{k}_{n-1} |\hat\phi_{n-k}|^2 \d \mu_{S^3}.
\end{eqnarray*}

\subsection{Energy fluxes in the partial conformal compactification}
We choose the vectors $\tau^b, \tau^c...\tau^f$ as follows
$$\tau^b = \tau^c = ... = \tau^f = \partial_t = \partial_u = \frac{1}{\sqrt 2} (\widetilde{n}^a + R^2 \widetilde{l}^a).$$
The unit normal vector to the hypersurface $\Sigma_0$ is
$$\widetilde{\nu}_{\Sigma_0}^a = r\partial_t = \frac{r}{\sqrt 2} (\widetilde{n}^a + R^2 \widetilde{l}^a).$$
We have 
\begin{eqnarray*}
\tau^b\tau^c...\tau^f &=& \left( \frac{1}{\sqrt 2} \right)^{n-1} (\widetilde{n}^b + R^2 \widetilde{l}^b ) (\widetilde {n}^c + R^2 \widetilde{l}^c)...(\widetilde{n}^f + R^2 \widetilde{l}^f) \cr
&=& \left( \frac{1}{\sqrt 2} \right)^{n-1} \left( \underbrace{\widetilde{n}^b \widetilde{n}^c...\widetilde{n}^f}_{n-1 \; terms} + R^{2(n-1)} \underbrace{\widetilde{l}^b \widetilde{l}^c...\widetilde{l}^f}_{n-1 \; terms} \right) \cr
&&+ \left( \frac{1}{\sqrt 2} \right)^{n-1} \left( \sum_{k=1}^{n-1}R^{2(k-1)}\underbrace{ \widetilde{l}^b...\widetilde{l}^c}_{k-1 \; terms} \underbrace{\widetilde{n}^d...\widetilde{n}^f}_{n-k \; terms}  + \sum_{k=1}^{n-1}R^{2k}\underbrace{ \widetilde{l}^b...\widetilde{l}^c \widetilde{l}^d}_{k \; terms} \underbrace{\widetilde{n}^e...\widetilde{n}^f}_{n-1-k \; terms} \right) \, ,
\end{eqnarray*}
and using the expression \eqref{phi},
\begin{eqnarray*}
\widetilde{\phi}_{AB...F}\bar{\widetilde{\phi}}_{A'B'...F'} &=& |\widetilde{\phi}_n|^2 \underbrace{\widetilde{l}_a...\widetilde{l}_f}_{n \; terms} + |\widetilde{\phi}_0|^2 \underbrace{\widetilde{n}_a...\widetilde{n}_f}_{n \; terms} \cr
&& + \sum_{k=1}^{n-1}|\widetilde{\phi}_{n-k}|^2 \left( \widetilde{n}_a \underbrace{\widetilde{n}_b...\widetilde{n}_c}_{k-1 \; terms} \underbrace{\widetilde{l}_d...\widetilde{l}_f}_{n-k \; terms}  + \widetilde{l}_a \underbrace{\widetilde{n}_b...\widetilde{n}_c\widetilde{n}_d}_{k \; terms} \underbrace{\widetilde{l}_e...\widetilde{l}_f}_{n-k-1 \; terms}  \right) + A,
\end{eqnarray*}
we obtain that
$$\widetilde{J}_a = \left( \frac{1}{\sqrt 2} \right)^{n-1} \left( \sum_{k=0}^{n-1}R^{2k}|\widetilde{\phi}_{n-k}|^2\widetilde{l}_a + \sum_{k=1}^n R^{2(k-1)}|\widetilde{\phi}_{n-k}|^2 \widetilde{n}_a     \right),$$
where $A$ vanishes due to the normalization condition of the Newman-Penrose tetrad. 

Using \eqref{e1} we can calculate that
\begin{eqnarray*}
\widetilde{J}_a \widetilde{\nu}_{\Sigma_0}^a &=& \left( \frac{1}{\sqrt 2} \right)^{n}r \left( \sum_{k=0}^{n-1}R^{2k}C_{n-1}^k|\widetilde{\phi}_{n-k}|^2\widetilde{l}_a + \sum_{k=1}^n R^{2(k-1)}C^{k-1}_{n-1}|\widetilde{\phi}_{n-k}|^2 \widetilde{n}_a \right) (\widetilde{n}^a + R^2 \widetilde{l}^a)\cr
&=& \left( \frac{1}{\sqrt 2} \right)^{n} \left( \sum_{k=0}^{n-1} R^{2k-1} C^k_{n-1} |\widetilde{\phi}_{n-k}|^2  + \sum_{k=1}^{n} R^{2k-1} C^{k-1}_{n-1} |\widetilde{\phi}_{n-k}|^2  \right),
\end{eqnarray*}
hence
\begin{eqnarray} \label{energy-sigma-2}
\widetilde{{\mathcal E}}_{\Sigma_0} (\widetilde{\phi}_{AB...F}) &=& \int_{\Sigma_0} \widetilde{J}_a \widetilde{{\nu}}_{\Sigma_0}^a (\widetilde{\nu}_{\Sigma_0}^a \hook \mathrm{dVol^4_{\widetilde{g}}}) = \int_{\Sigma_0} \widetilde{J}_a \widetilde{{\nu}}_{\Sigma_0}^a (r\partial_t \hook \mathrm{dVol^4_{\widetilde{g}}}) \cr
&=& \left( \frac{1}{\sqrt 2} \right)^n\int_{\Sigma_0} \left( \sum_{k=0}^{n-1} R^{2k} C^k_{n-1} |\widetilde{ \phi}_{n-k}|^2  + \sum_{k=1}^{n} R^{2k} C^{k-1}_{n-1} |\widetilde{\phi}_{n-k}|^2  \right) \d r \d^2\omega.
\end{eqnarray}
The normal vector to the null infinity hypersurface $\scri_T^+ = \scri^+ \cap \left\{ u \leq T \right\}$ is
$$\widetilde{{\mathcal{N}}}_{\scri_T^+}^a = \widetilde{n}^a = \sqrt {2} \partial_u,$$
hence the transversal vector to the null infinity hypersurface $\scri_T^+ = \scri^+ \cap \left\{ u \leq T \right\}$ is
$$\widetilde{\mathcal{L}}_{\scri_T^+}^a = -\frac{1}{\sqrt{2}} \partial_R.$$
With the supported compact initial data on $\Sigma_0$, the solution $\widetilde{\phi}_{AB...F}$ has the support on $\scri^+$ far away from $i_0$ (that is a consequence of the finite propagation speed). Therefore, we can calculate that
\begin{eqnarray*}
\widetilde{J}_a \widetilde{{\mathcal{N}}}_{\scri_T^+}^a &=& \left( \frac{1}{\sqrt 2} \right)^{n-1} \left( \sum_{k=0}^{n-1}R^{2k}|\widetilde{\phi}_{n-k}|^2\widetilde{l}_a + \sum_{k=1}^n R^{2(k-1)}|\widetilde{\phi}_{n-k}|^2 \widetilde{n}_a \right) \widetilde{n}^a \cr
&=& \left( \frac{1}{\sqrt 2} \right)^{n-1} |\widetilde{\phi}_n|^2 \; \mbox{(due to on $\scri_T^+$ we have $R^{2k} = 0 \; \mbox{with} \; k\geq 1$)}.
\end{eqnarray*}
Using \eqref{e} we get
\begin{eqnarray*}
\widetilde{{\mathcal E}}_{\scri^+_T} (\widetilde{\phi}_{AB...F}) &=& \int_{\scri^+_T}\widetilde{J}_a \widetilde{\mathcal{N}}_{\scri_T^+}^a (\widetilde{{\mathcal{L}}}_{\scri_T^+}^a \hook \mathrm{dVol^4_{\widetilde{g}}}) \cr
&=& \left( \frac{1}{\sqrt 2} \right)^{n} \int_{\scri^+_T} |\widetilde{\phi}_n|^2 \d u \d^2\omega,
\end{eqnarray*}
and we can define
\begin{equation}
\widetilde{\mathcal E}_{\scri^+} (\widetilde{\phi}_{AB...F}) = \lim_{T\rightarrow +\infty}\widetilde{{\mathcal E}}_{\scri^+_T} (\widetilde{\phi}_{AB...F}).
\end{equation}

\begin{figure}[H]
\begin{center}
\includegraphics[scale=0.3]{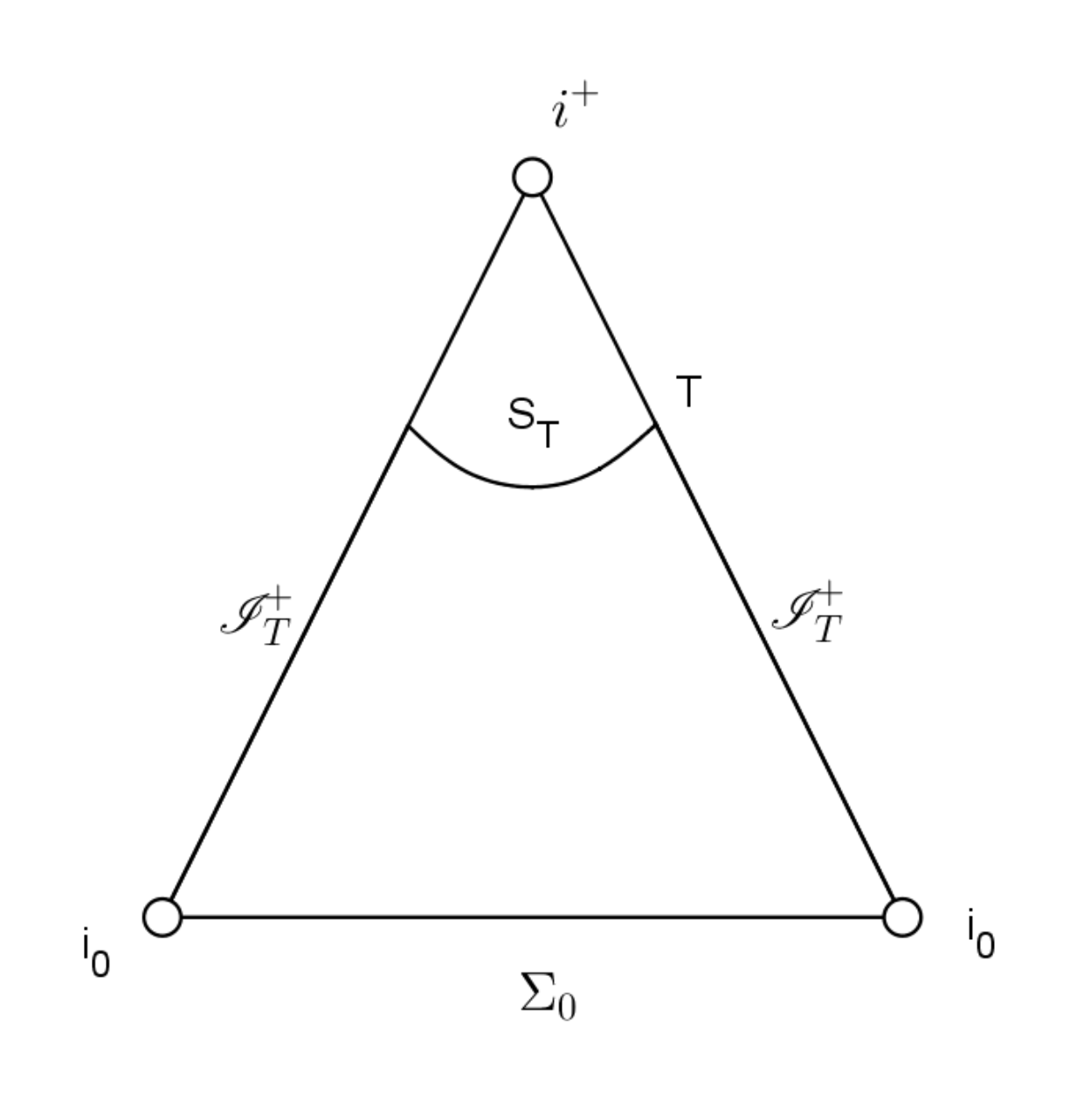}
\caption{The spacelike hypersurface $S_T$ in the partial conformal compactification spacetime $\widetilde{\mathbb{M}}$.}
\end{center}
\end{figure}
To prove the energy equality in the partial conformal compactification, we define a Cauchy (hyperboloid) hypersurface
$$S_T = \left\{ (t,r,\omega) \in {\mathbb R}_t \times {\mathbb R}_r \times S^2 \; ; \; t = T + \sqrt{1+r^2}  \right\},$$
which tends to $i^+$ as $T$ tends to $+\infty$. Since the initial data has a compact support on $\Sigma_0$, we obtain a closed form of the hypersurfaces $\Sigma_0, \scri_T^+$ and $S_T$. 

Now the conormal vector to the hypersurface $S_T$ is
$$\widetilde{{\mathcal N}}_a \d x^a = \d t - \frac{r}{\sqrt {1+r^2}} \d r.$$
Hence the unit normal vector to $S_T$ is
\begin{align}
\widetilde{{\mathcal N}}^a \partial x^a &= r^2 \left( \partial_t + \frac{r}{\sqrt {1+r^2}} \partial_r \right) = r^2\partial_t - \frac{r}{\sqrt{1+r^2}}\partial_R \cr
&= \frac{r^2}{\sqrt 2} \left( R^2\widetilde{l}^a + \widetilde{n}^a \right) + \frac{r\sqrt 2}{\sqrt{1+r^2}} \widetilde{l}^a \cr
&= \left( \frac{1}{\sqrt 2} + \frac{r\sqrt 2}{\sqrt{1+r^2}} \right) \widetilde{l}^a + \frac{r^2}{\sqrt 2}\widetilde{n}^a.
\end{align}
The transversal vector $\widetilde{{\mathcal L}}^a$ sastifies $\left < \widetilde{{\mathcal N}}^a,\widetilde{{\mathcal L}}^a \right > = 1$ that is 
$$\widetilde{{\mathcal L}}^a = \frac{1+r^2}{1+2r^2} \left( \partial_t - \frac{r}{\sqrt {1+r^2}} \partial_r \right).$$
Then, the contraction of $\widetilde{{\mathcal L}}^a$ with the volume form for $\widetilde{g}$ is
$$\widetilde{{\mathcal L}}^a \hook \mathrm{dVol}^4_{\widetilde{g}} = \frac{1+r^2}{1+2r^2}R^2\left( \d r + \frac{r}{\sqrt{1+r^2}}\d t\right) \d^2\omega.$$

On $S_T$ we have
$$\d t = \frac{r}{\sqrt{1+r^2}}\d r$$
hence
$$\widetilde{{\mathcal L}}^a \hook \mathrm{dVol}^4_{\widetilde{g}} = \frac{1+r^2}{1+2r^2}R^2 \left(1+ \frac{r^2}{1+r^2} \right)\d r\d^2\omega = R^2 \d r \d^2\omega.$$
Now we can calculate
\begin{eqnarray*}
\tau^b\tau^c...\tau^f &=& \left( \frac{1}{\sqrt 2} \right)^{n-1} (\widetilde{n}^b + R^2\widetilde{l}^b ) (\widetilde {n}^c + R^2\widetilde{l}^c)...(\widetilde{n}^f + R^2 \widetilde{l}^f) \cr
&=& \left( \frac{1}{\sqrt 2} \right)^{n-1} \left( \underbrace{\widetilde{n}^b \widetilde{n}^c...\widetilde{n}^f}_{n-1 \; terms} + R^{2(n-1)} \underbrace{\widetilde{l}^b \widetilde{l}^c...\widetilde{l}^f}_{n-1 \; terms} \right) \cr
&&+ \left( \frac{1}{\sqrt 2} \right)^{n-1} \left( \sum_{k=1}^{n-1}R^{2(k-1)}\underbrace{ \widetilde{l}^b...\widetilde{l}^c}_{k-1 \; terms} \underbrace{\widetilde{n}^d...\widetilde{n}^f}_{n-k \; terms}  + \sum_{k=1}^{n-1}R^{2k}\underbrace{ \widetilde{l}^b...\widetilde{l}^c \widetilde{l}^d}_{k \; terms} \underbrace{\widetilde{n}^e...\widetilde{n}^f}_{n-1-k \; terms} \right).
\end{eqnarray*}
Combining with Formula \eqref{phi} we get
\begin{eqnarray*}
\widetilde{\phi}_{AB...F}\bar{\widetilde{\phi}}_{A'B'...F'} &=& |\widetilde{\phi}_n|^2 \underbrace{\widetilde{l}^a...\widetilde{l}^f}_{n \; terms} + |\widetilde{\phi}_0|^2 \underbrace{\widetilde{n}^a...\widetilde{n}^f}_{n \; terms} \cr
&& + \sum_{k=1}^{n-1}|\widetilde{\phi}_{n-k}|^2 \left( \widetilde{n}^a \underbrace{\widetilde{n}^b...\widetilde{n}^c}_{k-1 \; terms} \underbrace{\widetilde{l}^d...\widetilde{l}^f}_{n-k \; terms}  + \widetilde{l}^a \underbrace{\widetilde{n}^b...\widetilde{n}^c\widetilde{n}^d}_{k \; terms} \underbrace{\widetilde{l}^e...\widetilde{l}^f}_{n-k-1 \; terms}  \right) + A,
\end{eqnarray*}
we obtain that
\begin{eqnarray*}
\left <\widetilde{J}^a, \widetilde{{\mathcal N}}^a \right >\widetilde{{\mathcal N}}^a &=& \widetilde{\phi}_{AB...F}\bar{\widetilde{\phi}}_{A'B'...F'}\tau^b\tau^c...\tau^f\widetilde{N}^a\cr
&=& \left( \frac{1}{\sqrt 2} \right)^{n-1} \left( \frac{1}{\sqrt 2} + \frac{r\sqrt 2}{\sqrt{1+r^2}} \right)\sum_{k=1}^{n}R^{2(k-1)}C_{n-1}^{k-1}|\widetilde{\phi}_{n-k}|^2 \cr
&&+ \left( \frac{1}{\sqrt 2} \right)^{n}r^2 \sum_{k=0}^{n-1}R^{2k}C_{n-1}^{k}|\widetilde{\phi}_{n-k}|^2,
\end{eqnarray*}
where $A$ vanishes due to the normalization condition of the Newman-Penrose tetrad. By associating with $\widetilde{{\mathcal L}}^a \hook \mathrm{dvol}^4_{\widetilde{g}}$ and by using \eqref{e} we calculate the energy of $\widetilde{\phi}_{AB...F}$ across $S_T$ as follows
\begin{eqnarray}\label{energy_on_St}
\widetilde{{\mathcal E}}_{S_T} (\widetilde{\phi}_{AB...F}) &=& \int_{S_T} \left <\widetilde{J}^a, \widetilde{{\mathcal N}}^a \right > (\tilde{{\mathcal L}}^a \hook \mathrm{dVol}^4_{\widetilde{g}}) \cr
&=&  \int_{S_T} \left( \frac{1}{\sqrt 2} \right)^{n-1} \left( \frac{1}{\sqrt 2} + \frac{r\sqrt 2}{\sqrt{1+r^2}} \right)\sum_{k=1}^{n}R^{2k}C_{n-1}^{k-1}|\widetilde{\phi}_{n-k}|^2 \d r \d^2\omega \cr
&&+  \int_{S_T} \left( \frac{1}{\sqrt 2} \right)^{n} \sum_{k=0}^{n-1}R^{2k}C_{n-1}^{k}|\widetilde{\phi}_{n-k}|^2\d r \d^2\omega.
\end{eqnarray}

\subsection{Energy equalities}
\begin{theorem}\label{egalite_energies}
In the conformal compactification spacetimes $\hat{\mathbb{M}}$ and $\widetilde{\mathbb{M}}$ we have energy equalities of the spin-$n/2$ zero rest-mass field through the hypersurfaces $\scri^\pm$, $S^3=\left\{ \tau = 0 \right\}$ in $\hat{\mathbb{M}}$ and $\Sigma_0 = \left\{ t=0 \right\}$ in $\widetilde{\mathbb{M}}$ as follows
\begin{equation}\label{eqen1}
\hat{\mathcal E}_{\scri^\pm}(\hat{\phi}_{AB...F}) = \hat{\mathcal E}_{S^3}({\hat\phi}_{AB...F}),
\end{equation}
\begin{equation}\label{eqen2}
\widetilde{{\mathcal E}}_{\scri^\pm}(\widetilde{\phi}_{AB...F}) = \widetilde{{\mathcal E}}_{\Sigma_0}(\widetilde{\phi}_{AB...F}).
\end{equation} 
\end{theorem}
\begin{proof}
Since the vector fields $\partial_t$ and $\partial_\tau$ are Killing, we can obtain the conservation laws
$$\hnabla^a\hat J_a = \widetilde{\nabla}^a\widetilde{J}_a = 0.$$
In $\hat{\mathbb{M}}$, we consider the domain $W$ which is formed by the hypersurfaces $S^3$ and $\scri^+$. Integrating the conservation law $\hnabla^a\hat J_a = 0$ on $W$ and by using the divergence theorem we obtain the energy equality \eqref{eqen1} between the energies of field through $\scri^+$ and $S^3$. The proof is similar for the energy of the field through $\scri^-$.  

In $\widetilde{\mathbb{M}}$, since $i^+$ is infinite, we consider the domain $\mathcal{W}$ which is formed by the hypersurfaces $\Sigma_0,\scri^+_T$ and $S_T$. Integrating the conservation law $\widetilde{\nabla}^a\widetilde{J}_a = 0$ on $\mathcal{W}$ and by using again the divergence theorem we obtain that
$$\widetilde{{\mathcal E}}_{\scri^+_T}(\widetilde{\phi}_{AB...F}) + \widetilde{{\mathcal E}}_{S_T}(\widetilde{\phi}_{AB...F}) = \widetilde{{\mathcal E}}_{\Sigma_0}(\widetilde{\phi}_{AB...F}).$$
Taking the limit as $T$ tend to $+\infty$ for the equality above, we get
$$\widetilde{{\mathcal E}}_{\scri^+}(\widetilde{\phi}_{AB...F}) + \lim_{T\rightarrow +\infty}\widetilde{{\mathcal E}}_{S_T}(\widetilde{\phi}_{AB...F}) = \widetilde{{\mathcal E}}_{\Sigma_0}(\widetilde{\phi}_{AB...F}) \; .$$
Since the pointwise decays of the components $\widetilde{\phi}_{n-k}$ obtained in Theorem \ref{decay} and the formula \eqref{energy_on_St} of the energy flux through $S_T$, we have
\begin{eqnarray*}
&&\lim_{T\rightarrow +\infty}\widetilde{{\mathcal E}}_{S_T}(\widetilde{\phi}_{AB...F}) \cr
&=& \lim_{T\rightarrow +\infty} \int_{S_T} \left( \frac{1}{\sqrt 2} \right)^{n-1} \left( \frac{1}{\sqrt 2} + \frac{r\sqrt 2}{\sqrt{1+r^2}} \right)\sum_{k=1}^{n}R^{2k}C_{n-1}^{k-1}|\widetilde{\phi}_{n-k}|^2 \d r \d^2\omega \cr
&&+ \lim_{T\rightarrow +\infty} \int_{S_T} \left( \frac{1}{\sqrt 2} \right)^{n} \sum_{k=0}^{n-1}R^{2k}C_{n-1}^{k}|\widetilde{\phi}_{n-k}|^2\d r \d^2\omega  \cr
&\leq & \lim_{T\rightarrow +\infty} \int_{S_T} \left( \frac{1}{\sqrt 2} \right)^{n-1} \left( \frac{1}{\sqrt 2} + \frac{r\sqrt 2}{\sqrt{1+r^2}} \right)\sum_{k=1}^{n}R^{2k}C_{n-1}^{k-1} \left|\frac{r^{k+1}}{t(r)^{n+2}} \right|^2 \d r \d^2\omega \cr
&& + \lim_{T\rightarrow +\infty} \int_{S_T} \left( \frac{1}{\sqrt 2} \right)^{n} \sum_{k=0}^{n-1}R^{2k}C_{n-1}^{k} \left|\frac{r^{k+1}}{t(r)^{n+2}} \right|^2\d r \d^2\omega  \cr
&=& \lim_{T\rightarrow +\infty} \int_{S_T} \left( \frac{1}{\sqrt 2} \right)^{n-1} \left( \frac{1}{\sqrt 2} + \frac{r\sqrt 2}{\sqrt{1+r^2}} \right)\sum_{k=1}^{n} C_{n-1}^{k-1} \left|\frac{r}{t(r)^{n+2}} \right|^2 \d r \d^2\omega \cr
&& + \lim_{T\rightarrow +\infty} \int_{S_T} \left( \frac{1}{\sqrt 2} \right)^{n} \sum_{k=0}^{n-1}C_{n-1}^{k} \left|\frac{r}{t(r)^{n+2}}\right|^2\d r \d^2\omega,
\end{eqnarray*}
where $t(r) = T + \sqrt{1+r^2}$. We can control the right-hand side of the inequality above as follows
\begin{eqnarray*}
\mbox{Right-hand side}&\leq& C \lim_{T\rightarrow +\infty}\int_{S_\omega^2} \int_{r=0}^{+\infty} \frac{1}{t(r)^4}  \d r \d \omega^2 \cr
&\leq& 2\pi C \lim_{T\rightarrow +\infty}\int_{r=0}^{+\infty} \frac{1}{t(r)^4} \d r  \cr
&=& 2\pi C \lim_{T\rightarrow +\infty} \left( \int_{r=0}^{T}  \frac{1}{t(r)^4}  \d r + \int_{r=T}^{+\infty}  \frac{1}{t(r)^4}  \d r \right) \cr
&\leq& 2\pi C \lim_{T\rightarrow +\infty} \left( \int_{r=0}^{T} \frac{1}{T^4} \d r + \int_{r=T}^{+\infty} \frac{1}{r^4}  \d r \right) \cr
&=& 2\pi C \lim_{T\rightarrow +\infty} \frac{2}{3T^3} = 0,
\end{eqnarray*}
due to $t(r) > r$ and $t(r) > T$. Therefore,
$$ \lim_{T\rightarrow +\infty}\widetilde{{\mathcal E}}_{S_T}(\widetilde{\phi}_{AB...F}) = 0.$$
Therefore we can obtain the energy equality \eqref{eqen2} in the partial comformal compactification spacetime $\widetilde{\mathbb{M}}$:
$$\widetilde{{\mathcal E}}_{\scri^+}(\widetilde{\phi}_{AB...F}) = \widetilde{{\mathcal E}}_{\Sigma_0}(\widetilde{\phi}_{AB...F}).$$
The proof is similar for the energy of field through $\scri^-$.
\end{proof}
\begin{cor}\label{equality_energy_2}
If we cut the partial conformal compactification $\widetilde{{\mathbb{M}}}$ by a spacelike hypersurface $\mathcal{S}$, suppose that $\mathcal{S}$ intersects $\scri^+$ at $Q$. Then we also have the energy equality
$$ \widetilde{{\mathcal E}}_{\scri^{+,Q}}(\widetilde{\phi}_{AB...F}) = \widetilde{{\mathcal E}}_{\mathcal{S}}(\widetilde{\phi}_{AB...F})  \, ,$$
where $\scri^{+,Q}$ is the future part of $Q$ in $\scri^+$.
\end{cor}

Another consequence of the energy equalities and the wellposedness of the Cauchy problem is that we can define the trace operator in the full and partial conformal compactification spacetimes due to the energy on the null infinity $\scri^+$ of the rescaled spin-$n/2$ zero rest-mass field being finite. The definition in the partial conformal compactification spacetime $\widetilde{\mathbb{M}}$ is as follows, the one in the full conformal compactification is similar.
\begin{defn}
The trace operator ${\mathcal T}^+: \mathcal{C}_0^{\infty}(\Sigma_0,\mathbb{S}_{(AB...F)}) \to \mathcal{C}_0^{\infty}(\scri^+,\mathbb{C})$ is given by
\begin{align*}
\mathcal{T}^+: \mathcal{C}_0^{\infty}(\Sigma_0,\mathbb{S}_{(AB...F)}) &\longrightarrow \mathcal{C}_0^{\infty}(\scri^+,\mathbb{C})\cr
\widetilde{{\psi}}_{AB...F}  &\longmapsto \widetilde{{\phi}}_n|_{\scri^+}.
\end{align*}
\end{defn}
Using again the energy equality we can extend the domain of the trace operator $\mathcal{T}^+$, where the extended operator is injective and has a closed range.
\begin{cor}\label{trace-one-one} 
We extend the trace operator 
\begin{align*}
\mathcal{T}^+: \mathcal{H}_0 = L^2(\Sigma_0, \mathbb{S}_{(AB...F)}) &\longrightarrow  \mathcal{H}^+ = L^2(\scri^+,\mathbb{C})\cr
\widetilde{{\phi}}_{AB...F}|_{\Sigma_0} &\longmapsto \widetilde{{\phi}}_n|_{\scri^+}
\end{align*}
where $\mathcal{H}_0 = L^2(\Sigma_0,\mathbb{S}_{(AB...F)})$ is the closed space of $\mathcal{C}_0^\infty(\Sigma_0,\mathbb{S}_{AB...F})$ in the energy norm 
\begin{align*}
\left\|\widetilde{{\phi}}_{AB...F}\right\|_{\Sigma_0}^2 &= \widetilde{{\mathcal E}}_{\Sigma_0} (\widetilde{\phi}_{AB...F}) \cr
&= \left( \frac{1}{\sqrt 2} \right)^n\int_{\Sigma_0} \left( \sum_{k=0}^{n-1} R^{2k} C^k_{n-1} |\widetilde{ \phi}_{n-k}|^2  + \sum_{k=1}^{n} R^{2k} C^{k-1}_{n-1} |\widetilde{\phi}_{n-k}|^2  \right) \d r \d^2\omega,
\end{align*}
and similarly $\mathcal{H}^+ = L^2(\scri^+,\mathbb{C})$ is the closed space of $\mathcal{C}_0^\infty(\scri^+,\mathbb{C})$ in the energy norm 
$$\left\|\widetilde{{\phi}}_n|_{\scri^+}\right\|_{\mathcal{H}^+}^2 = \widetilde{{\mathcal E}}_{\scri^+} (\widetilde{\phi}_{AB...F}) = \left( \frac{1}{\sqrt 2} \right)^{n} \int_{\scri^+} |\widetilde{\phi}_n|^2 \d u \d^2\omega.$$
The trace operator in the new domains is injective and has a closed range.
\end{cor}
\begin{proof}
We observe that $\mathcal{T}^+: \mathcal{H}_0 \to \mathcal{H}^+$ is injective from the energy equality in $\widetilde{\mathbb{M}}$. By using again the energy equality we have $\mathcal{T}^+$ transforms a Cauchy sequence to another one. Hence, the domain image $\mathcal{T}^+(L^2(\Sigma_0,\mathbb{S}_{(AB...F)}))$ is closed.
\end{proof}

\section{Goursat problem}\label{S6}
Since the scalar curvature is $\mathrm{Scal}_{\hat{g}}=6$ in $\hat{\mathbb{M}}$, the curvature spinors do not vanish in $\hat{\mathbb{M}}$. For convenience, we consider the Gousat problem in the partial conformal compactification spacetime $\widetilde{{\mathbb{M}}}$ which has $\mathrm{Scal}_{\widetilde{g}}=0$:
\begin{align}\label{Goursat1}
\begin{cases}
\widetilde{\nabla}^{AA'} \widetilde{\phi}_{AB...F} &= 0,\cr
\widetilde{\phi}_n|_{\scri^+} &= \widetilde{\psi}_n \in \mathcal{C}_0^\infty(\scri^+, {\mathbb C}),\cr
\widetilde{\phi}_{AB...F}|_{\scri^+} = \widetilde{\psi}_{AB...F} \in \mathcal{D}_{\scri^+}
\end{cases}
\end{align}
here $\mathcal{D}_{\scri^+}$ is the constraint space on $\scri^+$.

We recall the expression of the massless equation $\widetilde{\nabla}^{AA'} \widetilde{\phi}_{AB...F} = 0$ in $\widetilde{{\mathbb M}}$ (see Equation \eqref{dirac-part}):
\begin{align*}
\begin{cases}
-\frac{1}{\sqrt 2}\partial_R \widetilde{\phi}_k - \frac{1}{\sqrt 2}\left( \partial_\theta - \frac{i}{\sin\theta}\partial_\varphi + (n-2k+2)\frac{\cot\theta}{2}  \right) \widetilde{\phi}_{k-1} &= 0,\cr
\left(\sqrt 2 \partial_u + \frac{R^2}{\sqrt 2}\partial_R - (n-2k)\frac{R}{\sqrt 2} \right) \widetilde{\phi}_k - \frac{1}{\sqrt 2}\left( \partial_\theta + \frac{i}{\sin\theta}\partial_\varphi  - (n-2k-2)\frac{\cot\theta}{2}  \right) \widetilde{\phi}_{k+1} &= 0, 
\end{cases}
\end{align*}
where $k=1,2...n$ in the first equation and $k=0,1...n-1$ in the second one. Since the constraint system on $\scri^+$ is the projection of the equation $\widetilde{\nabla}^{AA'}\widetilde{{\phi}}_{AB...F} = 0$ on the null normal vector $\widetilde{n}^a$, the constraint on the null infinity hypersurface $\scri^+$ is that of the second equation of the system above on $\scri^+$:
\begin{equation*}
\sqrt 2 \partial_u \widetilde{\phi}_k|_{\scri^+} - \frac{1}{\sqrt 2}\left( \partial_\theta + \frac{i}{\sin\theta}\partial_\varphi  - (n-2k-2)\frac{\cot\theta}{2}  \right) \widetilde{\phi}_{k+1}|_{\scri^+} = 0.
\end{equation*}
Therefore on $\scri^+$, we have
\begin{equation*}
\widetilde{\phi}_k|_{\scri^+}(u) = \widetilde{\phi}_k|_{\scri^+}(-\infty) + \frac{1}{2}\int^{u}_{-\infty}\left( \partial_\theta + \frac{i}{\sin\theta}\partial_\varphi  - (n-2k-2)\frac{\cot\theta}{2}  \right)\widetilde{\phi}_{k+1}|_{\scri^+}(s) \d s,  
\end{equation*}
where $k=0,1...n-1$. This shows that from the initial data $\widetilde{\psi}_n \in \mathcal{C}^\infty_0({\scri^+}, \mathbb C)$ (its support away from $i^+$ and $i^0$), we can find the other components to obtain the full spinor field $\widetilde{\psi}_{AB...F}:=\widetilde{\phi}_{AB...F}|_{\scri^+}$. And we can think that its support is far away from $i^+$.

Let ${\mathcal {S}}$ be a spacelike hypersurface in $\widetilde{{\mathbb M}}$ such that it crosses $\scri^+$ strictly in the past of the support of the initial data $\widetilde{\psi}_n$. We denote the point of intersection of $\mathcal{S}$ and $\scri^+$ by $Q$. 

\begin{figure}[H]
\begin{center}
\includegraphics[scale=0.5]{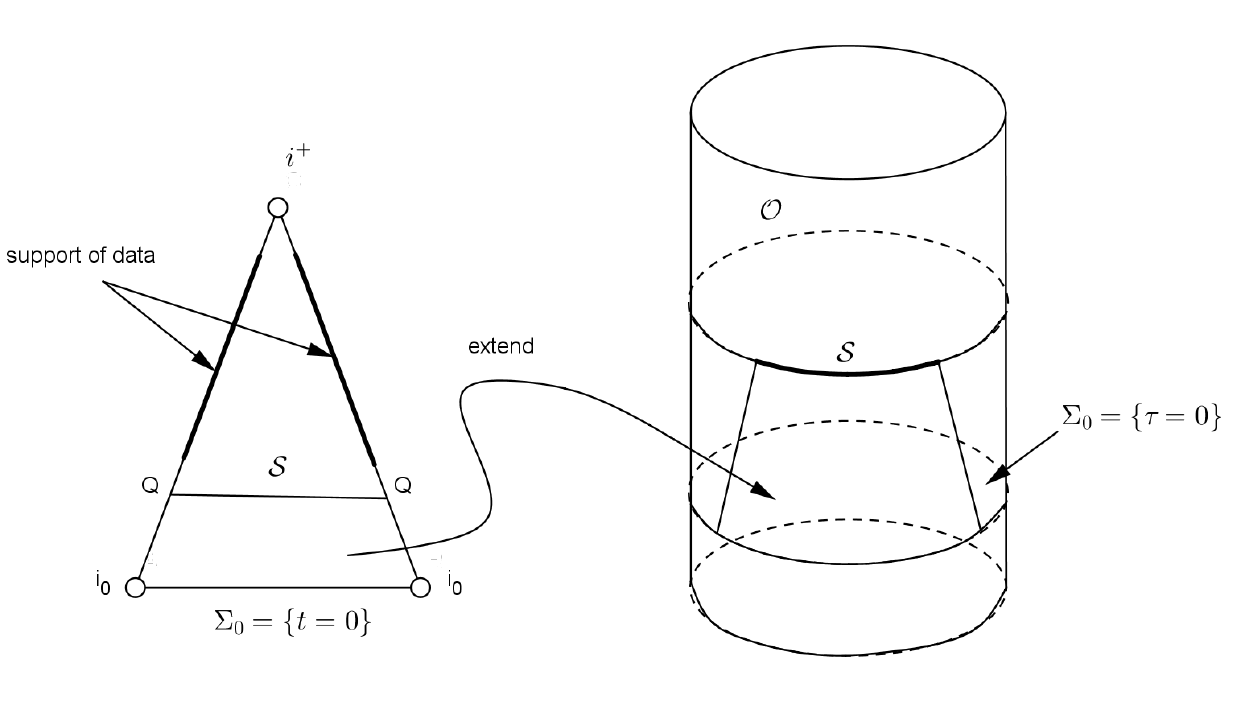}
\caption{Embedding the domain $\mathcal{I}^-(\mathcal{S})$ into Einstein's cylinder.}
\end{center}
\end{figure}

The Goursat problem will be established by the following two steps:

{\bf Step one:} We establish the wellposedness of the Goursat problem in the future $\mathcal {I}^+(\mathcal{S})$ of ${\mathcal S}$. On the partial compactification we have (see equation \eqref{app_3_wave0} in Appendix \ref{app_3_Commutator})
$$2\widetilde{\nabla}_{ZA'}\widetilde{\nabla}^{AA'} \widetilde{\phi}_{AB...F} = \widetilde{{\Box}}\widetilde{\phi}_{ZB...F} = 0 .$$
Therefore, the Goursat problem on the future $\mathcal{I}^+(\mathcal{S})$ has a problem consequence as follows
\begin{align}\label{Goursat2}
\begin{cases}
\widetilde{\Box} \widetilde{\phi}_{AB...F} &= 0,\cr
\widetilde{\phi}_{AB...F}|_{\scri^{+,Q}} &= \widetilde{\psi}_{AB...F}|_{\scri^{+,Q}} \in \mathcal{C}^\infty_0(\scri^{+,Q}, {\mathbb S}_{(AB...F)}),\cr
\widetilde{\nabla}^{AA'}\widetilde{\phi}_{AB...F}|_{\scri^{+,Q}} & = 0.
\end{cases}
\end{align}
where $\scri^{+,Q}$ is the future part of $Q$ in the null infinity hypersurface $\scri^+$. Here, we apply the generalisation of H\"ormander's result (see Appendix \ref{app_3_Goursat}), the system \eqref{Goursat2} has a unique solution.

Now we show that the solution of \eqref{Goursat2} is also a solution of the system \eqref{Goursat1} by proving that $\widetilde{\nabla}^{AA'} \widetilde{\psi}_{AB...F} = 0$ and using again the generalisation of H\"ormander's result. First, the components of $\widetilde{\nabla}^{AA'} \widetilde{\psi}_{AB...F}$ i.e the restrictions of the components of $\widetilde{\nabla}^{AA'} \widetilde{\phi}_{AB...F}$ on the hypersurface $\scri^{+,Q}$ are both zero. Indeed, if we set
$$\Xi^{A'}{_{B...F}}:= \widetilde{\nabla}^{AA'} \widetilde{\phi}_{AB...F},$$
then $\Xi^{A'}{_{B...F}}$ is symmetric in the indicies $B...F$, and we have
$$\Xi^{1'}{_{B...F}}|_{\scri^{+,Q}} = \widetilde{\iota}_{A'}\Xi^{A'}{_{B...F}}|_{\scri^{+,Q}} = \widetilde{\nabla}^{AA'} \widetilde{\phi}_{AB...F}|_{\scri^{+,Q}} = 0,$$
hence all the components of $\Xi^{1'}{_{B...F}}|_{\scri^{+,Q}}$ on $\scri^{+,Q}$ are zero. For the components of $\Xi^{0'}{_{B...F}}|_{\scri^{+,Q}}$, by the equation 
$$\widetilde{\Box} \widetilde{\phi}_{AB...F} = \frac{1}{2} \widetilde{\nabla}_{AK'}\widetilde{\nabla}^{KK'} \widetilde{\phi}_{KB...F} = \frac{1}{2} \widetilde{\nabla}_{AK'}\Xi^{K'}{_{B...F}} = \frac{1}{2} \Theta_{AB...F} =  0.$$
We have
$$\Theta_{1B...F} = \Theta_{0B...F} = 0,$$
where $\Theta_{1B...F}$ and $\Theta_{0B...F}$ are obtained by the differential equations which are of order one in the components of $\Xi^{1'}{_{B...F}}$ and $\Xi^{0'}{_{B...F}}$ (for detail see Appendix \ref{app_3_express}). Taking the constraint of these equations on $\scri^{+,Q}$ we obtain the restrictive equations of the components of $\Xi^{1'}{_{B...F}}$ and $\Xi^{0'}{_{B...F}}$ on $\scri^{+,Q}$. Since all the components of $\Xi^{1'}{_{B...F}}$ are zero on $\scri^{+,Q}$, we can obtain the Cauchy problem of the system of differential equations of order one, where the unknowns are only the restrictions of the components of $\Xi^{0'}{_{B...F}}$ on $\scri^{+,Q}$:
\begin{align}
\begin{cases}
\Theta_{1B...F}|_{\scri^+} &= 0,\cr
\Xi^{0'}{_{B...F}}|_{{\mathcal V}(P)} &= 0 
\end{cases}
\end{align}
where ${\mathcal V}(P)$ is the neighborhood of the point $P$ chosen to belong to $\scri^{+,Q}$, near $i^+$ and not belonging to the support of $\widetilde{\psi}_{AB...F}$. Since the Cauchy problem has a unique solution, we have the components of $\Xi^{0'}{_{B...F}}|_{\scri^{+,Q}}$ are also both zero. Therefore we have that the restrictions of the components of $\widetilde{\nabla}^{AA'} \widetilde{\phi}_{AB...F}$ on $\scri^{+,Q}$ are both zero (see Appendix \ref{app_3_express}).

Now we have (see Equation \eqref{app_3_wave1} in Appendix)
$$0 = \widetilde{\nabla}^{AA'}\widetilde{\Box}\widetilde{\phi}_{AB...F} = \frac{1}{2} \widetilde{\nabla}^{AA'}\widetilde{\nabla}_{AK'}\Xi^{K'}{_{B...F}} = \frac{1}{4} \widetilde{\Box}\Xi^{A'}{_{B...F}} + \frac{1}{2} \widetilde{\Box}^{A'}{_{K'}} \Xi^{K'}{_{B...F}},$$
raising the indicies $B...F$, we obtain the system
\begin{align}
\begin{cases}
\widetilde{\Box}\Xi^{A'B...F} + \widetilde{\Box}^{A'}{_{K'}} \Xi^{A'B...F}&= 0,\cr
\mbox{The restrictions of all the components of} \; \Xi^{A'B...F} \; \mbox{on} \; \scri^{+,Q} &= 0 
\end{cases}
\end{align}
with
$$\widetilde{\Box}^{A'}{_{K'}} \Xi^{K'B...F}= \widetilde{{\bar X}}^{A'}{_{K'Q'}}{^{K'}} \Xi^{Q'B...F} + \widetilde{\Phi}^{A'}{_{K'Q}}{^B} \Xi^{K'Q...F} + ... + \widetilde{\Phi}^{A'}{_{K'Q}}{^F} \Xi^{K'B...Q},$$
where $\widetilde{X}_{ABCD}$ and $\widetilde{\Phi}_{ABC'D'}$ are the curvature spinor. We have the rescaled scalar curvature $\mathrm{Scal}_{\tilde{g}} = \widetilde{\Lambda} = 0$, the Weyl spinor is conformal invariant $\widetilde{\Psi}_{ABCD} = \Psi_{ABCD}$ (see Appendix \ref{app_1}) and $\Psi_{ABCD} = 0$ in Minkowski spacetime $\mathbb{M}$. Therefore
$$\widetilde{X}_{ABCD} = \widetilde{\Psi}_{ABCD} + \widetilde{\Lambda} (\widetilde{\varepsilon}_{AC}\widetilde{\varepsilon}_{BD} + \widetilde{\varepsilon}_{AD}\widetilde{\varepsilon}_{BC})= {\Psi}_{ABCD} + \widetilde{\Lambda} (\widetilde{\varepsilon}_{AC}\widetilde{\varepsilon}_{BD} + \widetilde{\varepsilon}_{AD}\widetilde{\varepsilon}_{BC}) = 0.$$
The components of $\widetilde{\Phi}_{ABC'D'}$ are $\mathcal {C}^{\infty}$ due to the following formula (see \cite[Lemma A.1]{MaNi2012} for the generalized formula in Schwarzschild spacetime)
$$\widetilde{\Phi}_{ab} \d x^a \d x^b = \frac{1}{2} \left( R^2 \d u^2 - 2 \d u \d R + \d \omega^2 \right).$$
Using the generalisation of H\"ormander's result (see Appendix \ref{app_3_Goursat}), we get $\Xi^{A'B...F} = 0$ and then $\Xi^{A'}{_{B...F}} = \widetilde{\nabla}^{AA'} \widetilde{\phi}_{AB...F} = 0$. So the solution of the system \eqref{Goursat2} is a solution of the system \eqref{Goursat1}. For convenience, we denote by $\widetilde{\phi}^1_{AB...F}$ the solution of this step.

{\bf Step two:} We need to extend the solution of the Goursat problem on future $\mathcal{I}^+(\mathcal{S})$ down to $\Sigma_0$. This is equivalent to prove the wellposedness of the Cauchy problem in the past $\mathcal{I}^-(\mathcal{S})$ of $\mathcal{S}$: 
\begin{align}\label{Cauchy2}
\begin{cases}
\widetilde{\nabla}^{AA'} \widetilde{\phi}_{AB...F} &= 0,\cr
\widetilde{\phi}_{AB...F}|_{\mathcal{S}} &= \widetilde{\phi}^1_{AB...F}|_{\mathcal{S}}.
\end{cases}
\end{align}

Since the conformal transformations, the domain $\mathcal{I}^-(\mathcal{S})$ can be embedded into Einstein's cylinder. We extend $\mathcal{S}$ to the spacelike hypersurface $\mathcal{O}$ and the initial data $\widetilde{\phi}_{AB...F}|_{\mathcal{S}}$ is zero in the rest of the support. Now we can consider the equivalent Cauchy problem
\begin{align}\label{Cauchy3}
\begin{cases}
{\hnabla}^{AA'} {\hat\phi}_{AB...F} &= 0,\cr
{\hat\phi}_{AB...F}|_{\mathcal{S}} &= {\hat \phi}^1_{AB...F}|_{\mathcal{S}},\cr
{\hat\phi}_{AB...F}|_{\mathcal{O}/\mathcal{S}} &= 0.
\end{cases}
\end{align}
As a consequence of Theorem \ref{cauchyproblem}, this Cauchy problem is well-posed, we denote its solution by $\hat{\phi}^2_{AB...F}$ and the solution of this step by $\widetilde{\phi}^2_{AB...F}$. Clearly, we can obtain by using the divergence theorem that
\begin{equation}
\hat{\mathcal{E}}_{\mathcal{S}}(\hat\phi_{AB...F}) = \hat{\mathcal{E}}_{\Sigma_0}(\hat\phi_{AB...F}),
\end{equation}
and since the energy  of the spin-$n/2$ zero rest-mass field is invariant under the conformal transformations, we have
$$\widetilde{{\mathcal{E}}}_{\mathcal{S}}(\widetilde{{\phi}}_{AB...F}) = \widetilde{{\mathcal{E}}}_{\Sigma_0}(\widetilde{\phi}_{AB...F}).$$
Using the energy equality (see Theorem \ref{egalite_energies} and Corollary \ref{equality_energy_2}),  we obtain that
$$\widetilde{{\mathcal{E}}}_{\scri^{+,Q}}(\widetilde{{\phi}}_{AB...F})= \widetilde{{\mathcal{E}}}_{\mathcal{S}}(\widetilde{{\phi}}_{AB...F}) = \widetilde{{\mathcal{E}}}_{\Sigma_0}(\widetilde{\phi}_{AB...F}).$$
Therefore the energy of the solution through the hypersurface $\Sigma_0$ is finite and we can define the trace operator on $\Sigma_0$ as the constraint of solution of Cauchy problem \eqref{Cauchy2} on $\Sigma_0$.

Finally, the solution of the Goursat problem is the union of the solutions of two step above
$$\widetilde{\phi}_{AB...F}= \begin{cases}\widetilde{\phi}^1_{AB...F} \mbox{ in the domain $\mathcal{I}^+(S)$},\\ 
\widetilde{\phi}^2_{AB...F}\mbox{ in the domain $\mathcal{I}^-(S)$}.
\end{cases}$$

We summarize everything that we have just done by the following theorem
\begin{theorem}(Goursat problem)\label{Goursatprob}
The Goursat problem for the rescaled massless equation $\widetilde{\nabla}^{AA'}\widetilde{\phi}_{AB...F} = 0$ in $\widetilde{{\mathbb M}}$ is well-posed i.e for any $\widetilde{\psi}_n \in \mathcal{C}_0^\infty(\scri^+)$ and $\widetilde{\psi}_{AB...F}\in \mathcal{D}_{\scri^+}$ there exists a unique $\widetilde{\phi}_{AB...F}$ solution of $\widetilde{\nabla}^{AA'}\widetilde{\phi}_{AB...F} = 0$ such that   
$$\widetilde{\phi}_{AB...F} \in {\mathcal C}^\infty(\widetilde{{\mathbb{M}}},\mathbb{S}_{(AB...F)}) \; ; \; \widetilde{\phi}_n|_{\scri^+} = \widetilde{\psi}_n \, \mbox{and} \, \widetilde{\phi}_{AB...F}|_{\scri^+} = \widetilde{\psi}_{AB...F}.$$
Furthermore, the energy norm of the constraint of the solution $\widetilde{\phi}_{AB...F}|_{\Sigma_0}$ on $\Sigma_0$ is finite.
\end{theorem}
As a direct consequence of the wellposedness of the Goursat problem the extension of trace operator ${\mathcal T}^+$ given in Corollary \ref{trace-one-one} is surjective. Combining with Corollary \ref{trace-one-one} we have that the extension of trace operator is an isometric operator from ${\mathcal H}_0 = L^2(\Sigma_0,\mathbb{S}_{(AB...F)})$ onto ${\mathcal H}^+ = L^2(\scri^+, \mathbb C)$. Therefore, we obtain a conformal scattering operator of the generalized spin-$n/2$ zero rest-mass equation which is an isometric operator in the following corollary
\begin{cor}\label{Con}
We define similarly ${\mathcal T}^+$ the isometric operators 
$${\mathcal T}^-: {\mathcal H}_0 = L^2(\Sigma_0,\mathbb{S}_{(AB...F)}) \longrightarrow {\mathcal H}^-=L^2(\scri^-, \mathbb C),$$
given by
$${\mathcal T}^-(\widetilde{\phi}_{AB...F}|_{\Sigma_0}) = \widetilde{\phi}_0|_{\scri^-},$$
where the definition of ${\mathcal H}^-=L^2(\scri^-, \mathbb C)$ is the same $\mathcal H^+$.
Then the scattering operator ${\mathcal W}: = {\mathcal T}^+\circ({\mathcal T}^{-})^{-1}$ is an isometric operator that associates the past scattering ${\mathcal H}^-$ data to the future scattering ${\mathcal H}^+$ data.
\end{cor}
\begin{proof}
The proof is evidently due to the isometric property of $\mathcal{T}^\pm$.
\end{proof}

\section{Appendix}\label{A}

\subsection{Curvature spinors}\label{app_1}
Let $({\cal M},g)$ be a spacetime which has a spin structure and is equipped with the Levi-Civitta connection, the Riemann tensor $R_{abcd}$ can be decomposed as follows (see Penrose and Rindler \cite[eq. (4.6.1), pp. 231, Vol. 1]{PeRi}): 
\begin{equation}\label{decomposition_Riemann}
 R_{abcd} = X_{ABCD} \, \varepsilon_{A'B'} \varepsilon_{C'D'} + \Phi_{ABC'D'} \, \varepsilon_{A'B'} \varepsilon_{CD} + \bar{\Phi}_{A'B'CD} \, \varepsilon_{AB} \varepsilon_{C'D'} + \bar{X}_{A'B'C'D'} \, \varepsilon_{AB} \varepsilon_{CD},
\end{equation} 
where $X_{ABCD}$ is a complete contraction of the Riemann tensor in its primed spinor indices
$$ X_{ABCD} = \frac{1}{4} R_{abcd} {\varepsilon}^{A'B'} {\varepsilon}^{C'D'},$$
and ${\Phi}_{ab} = {\Phi}_{(ab)}$ is the trace-free part of the Ricci tensor multiplied by $-1/2$~:
$$2{\Phi}_{ab} = 6 {\Lambda} {g}_{ab} - {R}_{ab} \, ,~ {\Lambda} = \frac{1}{24} \mathrm{Scal}_{g}.$$

We set
$$P_{ab} = \Phi_{ab} - \Lambda g_{ab},$$
$$ X_{ABCD} = \Psi_{ABCD} + \Lambda \left( \varepsilon_{AC} \varepsilon_{BD} + \varepsilon_{AD} \varepsilon_{BC} \right),\, \Psi_{ABCD} = X_{(ABCD)} = X_{A(BCD)}.$$

Under a conformal rescaling $\hat g=\Omega^2 g$ we have (see Penrose and Rindler \cite[pp. 123, Vol. 2]{PeRi}):
\begin{gather*}
\hat{\Psi}_{ABCD} = \Psi_{ABCD}, \cr
\hat{\Lambda} = \Omega^{-2} \Lambda + \frac{1}{4} \Omega^{-3} \square \Omega, ~\square = \nabla^a \nabla_a, \cr
\hat{P}_{ab} = P_{ab} - \nabla_b \Upsilon_a  + \Upsilon_{AB'} \Upsilon_{BA'},~ \mbox{with } \Upsilon_a = \Omega^{-1} \nabla_a \Omega = \nabla_a \log \Omega.
\end{gather*}

\subsection{Spinor form of commutators}\label{app_3_Commutator}
In this section, we will give the spinor form of the commutators $\Delta^{ab} = \nabla^{[a}\nabla^{b]}$ (see \cite[pp. 242-244, Vol. 1]{PeRi} for the spinor form of $\Delta_{ab} = \nabla_{[a}\nabla_{b]}$) and establish some formulas used to prove the wellposedness of Cauchy and Goursat problems.

Since the anti-symmetric property of $\Delta^{ab}$, we have
$$\Delta^{ab} = 2\nabla^{[a}\nabla^{b]} = \varepsilon^{A'B'} \Box^{AB} + \varepsilon^{AB} \Box^{A'B'},$$
where
$$\Box^{AB} = \nabla^{X'(A}\nabla^{B)}{_{X'}}, \; \Box^{A'B'} = \nabla^{X(A'}\nabla^{B')}{_{X}}.$$
Now we have
$$\Delta^{ab} = g^{ac}g^{bd} \Delta_{cd},$$
and $\Delta_{ab}$ acts on the spinor form $\kappa^C $ as
$$\Delta_{ab}\kappa^C = \left\{ \varepsilon_{A'B'} X_{ABE}{^C} + \varepsilon_{AB} \Phi_{A'B'E}{^{C}} \right\} \kappa^E,$$
where $X_{ABCD}$ and $\Phi_{ABC'D'}$ are the curvature spinors in the expression of the Riemann tensor $R_{abcd}$~:
$$R_{abcd} = X_{ABCD} \, \varepsilon_{A'B'} \varepsilon_{C'D'} + \Phi_{ABC'D'} \, \varepsilon_{A'B'} \varepsilon_{CD} + \bar{\Phi}_{A'B'CD} \, \varepsilon_{AB} \varepsilon_{C'D'} + \bar{X}_{A'B'C'D'} \, \varepsilon_{AB} \varepsilon_{CD}.$$
Hence, we obtain
$$\Delta^{ab} \kappa^C = \varepsilon^{AC}\varepsilon^{A'C'} \varepsilon^{BD} \varepsilon^{B'D'}\Delta_{cd}\kappa^C = \left\{ \varepsilon^{A'B'} X^{AB}{_E}{^C} + \varepsilon^{AB} \Phi^{A'B'}{_E}{^{C}} \right\} \kappa^E.$$
By symmetrizing and skew-symmetrizing over $AB$, we get
$$\Box^{AB} \kappa^C =  X^{AB}{_E}{^C}  \kappa^E, \; \Box^{A'B'}\kappa^{C} =  \Phi^{A'B'}{_E}{^{C}}\kappa^E. $$
Similarly, we can obtain the formula of the primed spin-vectors
$$\Delta^{ab} \tau^{C '}= \left\{ \varepsilon^{AB} {\bar X}^{A'B'}{_{E'}}{^{C'}} + \varepsilon^{A'B'} \Phi^{AB}{_{E'}}{^{C'}} \right\} \tau{^{E'}},$$
$$\Box^{AB} \tau^{C'} =  \Phi^{AB}{_{E'}}{^{C'}}  \tau^{E'}, \; \Box^{A'B'}\tau^{C'} =  {\bar X}^{A'B'}{_{E'}}{^{C'}}\tau^{E'}. $$
Lowering the index $C$ (or $C'$), we also get
$$\Box^{AB} \kappa_C =   X^{ABE}{_C}  \kappa_E, \; \Box^{A'B'}\kappa_{C} =  \Phi^{A'B'E}{_{C}}\kappa^E, $$
$$\Box^{AB} \tau_{C'} =  \Phi^{AB{E'}}{_{C'}}  \tau_{E'}, \; \Box^{A'B'}\tau_{C'} =  {\bar X}^{A'B'{E'}}{_{C'}}\tau_{E'}. $$
For the actions of $\Box^{AB}$ and $\Box^{A'B'}$ on higher spin fields, we expand them by a sum of outer products of spin vectors and use the properties above. 

Now we establish the formulas which were used in the proofs of the Cauchy and Goursat problems. Frist, for the formulas in the Goursat problem we have
\begin{eqnarray}\label{app_3_wave0}
\widetilde{\nabla}_{ZA'}\widetilde{\nabla}^{AA'} \widetilde{\phi}_{AB...F} &=& \widetilde{\varepsilon}^{AM}\widetilde{\nabla}_{ZA'}\widetilde{\nabla}_M^{A'} \widetilde{\phi}_{AB...F} = \widetilde{\varepsilon}^{AM} \left( \widetilde{\nabla}_{A'[Z}\widetilde{\nabla}{_{M]}}{^{A'}} + \widetilde{\nabla}_{A'(Z}\widetilde{\nabla}{_{M)}}{^{A'}}  \right) \widetilde{\phi}_{AB...F} \cr
&=& \widetilde{\varepsilon}^{AM} \left( \frac{1}{2}\widetilde{\varepsilon}_{ZM} \widetilde{\Box} + \widetilde{\Box}_{ZM}  \right) \widetilde{\phi}_{AB...F} = \frac{1}{2} \widetilde{\varepsilon}_Z{^A}\widetilde{\Box} \widetilde{\phi}_{AB...F} + \widetilde{\Box}_{Z}{^A} \widetilde{\phi}_{AB...F} \cr
&=& \frac{1}{2}\widetilde{\Box} \widetilde{\phi}_{ZB...F} + \widetilde{X}_{ZA}{^{NA}} \widetilde{\phi}_{NB...F} - \widetilde{X}_Z {^A}{_B}{^N}\widetilde{\phi}_{AN...F} - ...\cr
&&- \widetilde{X}_Z {^A}{_F}{^N}\widetilde{\phi}_{AB...N} \cr
&=& \frac{1}{2}\widetilde{\Box}\widetilde{\phi}_{ZB...F},
\end{eqnarray}
due to the curvature spinors $\widetilde{X}_{ABCD}$ vanish in $\widetilde{\mathbb{M}}$:
$$\widetilde{X}_{ABCD} = \widetilde{\Psi}_{ABCD} + \widetilde{\Lambda} (\widetilde{\varepsilon}_{AC} \widetilde{\varepsilon}_{BD} + \widetilde{\varepsilon}_{AD}\widetilde{\varepsilon}_{BC}) = \Psi_{ABCD}+ \widetilde{\Lambda} (\widetilde{\varepsilon}_{AC} \widetilde{\varepsilon}_{BD} + \widetilde{\varepsilon}_{AD}\widetilde{\varepsilon}_{BC}) = 0.$$
We have also
\begin{eqnarray}\label{app_3_wave1}
&&\widetilde{\nabla}^{AA'}\widetilde{\nabla}_{AK'}\Xi^{K'}{_{B...F}} = - \widetilde{\varepsilon}_{K'M'}\widetilde{\nabla}^{AA'}\widetilde{\nabla}^{M'}_{A}\Xi^{K'}{_{B...F}} \cr
&=& - \widetilde{\varepsilon}_{K'M'} \left( \widetilde{\nabla}^{A[A'}\widetilde{\nabla}^{M']}_{A} + \widetilde{\nabla}^{A(A'}\widetilde{\nabla}^{M')}_{A} \right)\Xi^{K'}{_{B...F}} \cr
&=& - \widetilde{\varepsilon}_{K'M'} \left(  \frac{1}{2} \widetilde{\varepsilon}^{A'M'} \widetilde{\Box} + \widetilde{\Box}^{A'M'} \right)\Xi^{K'}{_{B...F}} \cr
&=& \frac{1}{2}\widetilde{\varepsilon}^{A'} {_{K'}}\widetilde{\Box}\Xi^{K'}{_{B...F}} + \widetilde{\Box}^{A'}{_{K'}}\Xi^{K'}{_{B...F}} \cr
&=& \frac{1}{2}\widetilde{\Box}\Xi^{A'}{_{B...F}} + \bar{{\widetilde X}}^{A'}{_{K'Q'}}{^{K'}} \Xi^{Q'}{_{B...F}} + \widetilde{\Phi}^{A'}{_{K'}}{^Q}{_B} \Xi^{K'}{_{Q...F}} + ... + \widetilde{\Phi}^{A'}{_{K'}}{^Q}{_F} {\Xi}^{K'}{_{B...Q}} \cr
&=& \frac{1}{2}\widetilde{\Box}\Xi^{A'}{_{B...F}} + \widetilde{\Phi}^{A'}{_{K'}}{^Q}{_B} \Xi^{K'}{_{Q...F}} + ... + \widetilde{\Phi}^{A'}{_{K'}}{^Q}{_F} {\Xi}^{K'}{_{B...Q}}.
\end{eqnarray}

The formula in the proof of the Cauchy problem is
\begin{eqnarray}\label{app_3_wave2}
\hnabla^{AA'}\hnabla_{A'}^Z {\hat\phi}_{ZAC...F} &=& \hnabla^{A'(A}\hnabla_{A'}^{Z)} {\hat\phi}_{ZAC...F} \cr
&=& \hat{\Box}^{AZ}\hat{\phi}_{AZC...F} \cr
&=& - \hat{X}^{AZM}{_A}\hat{\phi}_{MZC...F} - \hat{X}^{AZM}{_Z}\hat{\phi}_{AMC...F} - \hat{X}^{AZM}{_C}\hat{\phi}_{AZM...F}-... - \cr
&&-\hat{X}^{AZM}{_F}\hat{\phi}_{AZC...M}\cr
&=& -(n-1)\hat{\phi}_{AZM(C...K}\hat{\Psi}_{F)}{^{AZM}},
\end{eqnarray}
due to $\hat{X}^{A(ZM)}{_A} = 0$ and $\hat{X}^{(AZM)}{_C} = \hat{\Psi}^{AZM}{_C}$.

Note that if we define the wave operator by using the spinor form as follows
\begin{equation}\label{WaveOperator1}
\Box = \varepsilon^{MN}\varepsilon_{M'N'}\nabla_M^{M'}\nabla_N^{N'} = \nabla_a\nabla^a,
\end{equation}
then we can obtain
\begin{eqnarray}\label{WaveOperator2}
\Box &=& \varepsilon^{MN}\nabla_{N'M}\nabla_N^{N'} = \varepsilon^{MN} \left( \nabla_{N'[M}\nabla_{N]}^{N'} + \nabla_{N'(M}\nabla_{N)}^{N'} \right) \cr
&=& \varepsilon^{MN} \left( \frac{1}{2}\varepsilon_{MN}\check{\Box} + \check{\Box}_{MN} \right) \cr
&=& \check{\Box} - \check{\Box}_{M}^M.
\end{eqnarray}
Similarly
\begin{equation}\label{WaveOperator3}
\Box =  \varepsilon_{M'N'} \left( - \frac{1}{2}\varepsilon^{M'N'}{\breve{\Box}} + {\breve{\Box}}^{M'N'} \right)  = -{\breve{\Box}} + {\breve{\Box}}^{M'}_{M'}.
\end{equation}
\begin{remark}
We notice that the spin-wave operators $\Box$, $\check{\Box}$ and $\breve{\Box}$ (which act on the full spinor field $\phi_{AB...F}$) and the scalar wave operator $\Box_g$ (which acts on the scalar spin components $(\phi_0,\phi_1,\phi_2...\phi_n)$ of the spinor field $\phi_{AB...F}$), that are of the same operator module the derivation terms of the order less than or equal one. The module's terms are regular and can be calculated by using the derivations of spin-frame \eqref{Modu1} and \eqref{Modu2}.
\end{remark}

\subsection{The constraint system and non-trivial solutions}\label{app_3_existe_constraint}
In order to establish the constraint system of spin-$n/2$ zero rest-mass equation we recall notations and formulas of the intrinsic space spinor derivatives in Einstein's cylinder $(\mathfrak{C},\hat{g})$. We refer \cite[Section 2.1]{ABJ} for the formulas on Minkowski spacetime and \cite[Section 2.2, page 5]{MaNi2004} for the ones on asymptotic simple spacetimes. We consider $\mathcal{T}^a = \sqrt{2}\partial_\tau$ the future-oriented normal vector field to the foliation $\left\{\tau = constant \right\}_{\tau\in \mathbb{R}}$ of $\mathfrak{C}$. We have
$$\hat{g} = \frac{1}{2}N^2 \d \tau^2 - h,$$
where $N=\sqrt{2}$ and $h = \sigma_{S^3}^2 = \d \zeta^2 + \sin^2\zeta\d\omega^2$.

The connection $\nabla_a = \nabla_{AA'}$ can be decomposed along $\mathcal{T}^a$ and $\left(\mathcal{T}^a\right)^\perp$ as follows
\begin{equation}\label{decomeq}
\hnabla_{AA'} = \frac{1}{2}\mathcal{T}_{AA'}\hnabla_{\mathcal{T}} + D_{AA'} = \frac{1}{2}\mathcal{T}_{AA'}\hnabla_{\mathcal{T}} + \mathcal{T}_{A'}^B D_{AB},
\end{equation}
where $\hnabla_{\mathcal{T}}= \mathcal{T}^a\hnabla_a$ is the covariant derivative along $\mathcal{T}^a$
and $D_{AA'} = -{h_{AA'}}^{BB'}\hnabla_{BB'}$ is the part of $\hnabla_{AA'}$ orthogonal to $\mathcal{T}^a$, i.e, $\mathcal{T}^{AA'}D_{AA'}=0$.

By taking projection of \eqref{decomeq} on $\mathcal{T}^a$ we get 
$$D_{AB} = \mathcal{T}^{A'}_BD_{AA'} = \mathcal{T}^{A'}_{(B}\hnabla_{A)A'}$$
which does not contain the time derivative. We have also that
\begin{equation}\label{InDer}
D^{AB} = \mathcal{T}^B_{A'}D^{AA'} = \mathcal{T}^{(B}_{A'}\hnabla^{A)A'} = \mathcal{T}^{BA'}\hnabla^A_{A'} = \mathcal{T}^{AA'}\hnabla^B_{A'} 
\end{equation}
is the projection of $\hnabla^{AA'}$ on $\mathcal{T}^a$. The operators $D_{AB}$ and $D^{AB}$ are intrinsic space spinor derivatives on $\left\{ \tau = constant \right\}$.

The constraint equation of $\hnabla^{AA'}\hat{\phi}_{ABC...F} =0$ on $\left\{ \tau = constant \right\}$ that does not constain the time derivative. Hence, it is the projection of $\hnabla^{AA'}\hat{\phi}_{ABC...F} =0$ on $\mathcal{T}^a$ and thus is
\begin{equation}\label{constrain}
D^{AB}\left( \hat{\phi}_{ABC...F}|_{\tau=constant} \right) = \left(\mathcal{T}^{AA'}\hnabla^B_{A'}\hat{\phi}_{ABC...F} \right)|_{\tau=constant}=0.
\end{equation}
This equation corresponds to system \eqref{CONE} in the term of spin components. Indeed, we decompose $\hnabla^{AA'}\hat{\phi}_{ABC...F}=0$ on the rescaled spin-frame $\left\{\hat{o}_A,\hat{\iota}_A\right\}$ and its dual conjugation $\left\{\bar{\hat{o}}^{A'},\bar{\hat{\iota}}^{A'}\right\}$ to obtain that (see \cite[pp. 23-26, Vol. 2]{PeRi}):
\begin{eqnarray}\label{P}
0 &=& \hnabla^{AA'} \hat{\phi}_{ABC...F}\cr
&=& -\left(\hat{l}^a\hnabla_a {\hat\phi}_n - \bar{\hat{m}}^a\hnabla_a \hat{\phi}_{n-1} + n\hat{\varepsilon}\hat{\phi}_n - (n-2)\hat{\alpha}\hat{\phi}_{n-1}  \right.\cr
&&\left. \hspace{6cm}+ (n-1)\hat{\lambda}\hat{\phi}_{n-2} - n\hat{\pi}\hat{\phi}_{n-1}+ \hat{\rho}\hat{\phi}_n \right)\bar{\hat\iota}^{A'}\hat{o}_B \hat{o}_C...\hat{o}_F\cr
&&+(-1)^{n-k-1}\left[ \hat{n}^a\hnabla_a {\hat\phi}_k - \hat{m}^a\hnabla_a{\hat\phi}_{k+1} -(n-2k)\hat{\gamma}\hat{\phi}_k + (n-2k-2)\hat{\beta}\hat{\phi}_{k+1}\right.\cr
&&\left. - (n-k-1)\hat{\sigma}\hat{\phi}_{k+2} + (n-k)\hat{\tau}\hat{\phi}_{k+1} + (k+1)\hat{\mu}\hat{\phi}_k - k\hat{\nu}\hat{\phi}_{k-1} \right]\bar{\hat o}^{A'}\hat{o}_B \underbrace{\hat{o}_C...\hat{o}_M}_{k-1\, terms} \underbrace{\hat{\iota}_N...\hat{\iota}_F}_{n-k-1 \, terms}\cr
&&-(-1)^{n-k-1}\left[ \hat{l}^a\hnabla_a {\hat\phi}_k - \bar{\hat m}^a\hnabla_a\hat{\phi}_{k-1} - (n-2k)\hat{\varepsilon}\hat{\phi}_k + (n-2k+2)\hat{\alpha}\hat{\phi}_{k-1}\right.\cr
&&\left. + (k-1)\hat{\lambda}\hat{\phi}_{k-2} - k\hat{\pi}\hat{\phi}_{k-1} - (n-k+1)\hat{\rho}\hat{\phi}_k + (n-k)\hat{\kappa}\hat{\phi}_{k+1} \right] \bar{\hat \iota}^{A'}\hat{\iota}_B \underbrace{\hat{o}_C...\hat{o}_M}_{k-1\, terms}\underbrace{\hat{\iota}_N...\hat{\iota}_F}_{n-k-1 \, terms}\cr
&&+(-1)^n\left( \hat{n}^a\hnabla_a {\hat\phi}_0 - \hat{m}^a\hnabla_a {\hat\phi}_1 -n\hat{\gamma}\hat{\phi}_0 + (n-2)\hat{\beta}\hat{\phi}_1 \right.\cr
&&\left. \hspace{6cm}- (n-1)\hat{\sigma}\hat{\phi}_2 + n\hat{\tau}\hat{\phi}_1 + \hat{\mu}\hat{\phi}_0 \right) \bar{\hat o}^{A'}\hat{\iota}_B\hat{\iota}_C...\hat{\iota}_F\cr
&=& -\left(\hat{l}^a\hnabla_a {\hat\phi}_n - \bar{\hat m}^a\hnabla_a \hat{\phi}_{n-1} + \frac{\cot\theta}{2\sqrt 2\sin\zeta} (n-2) {\hat\phi}_{n-1} - \frac{\cos\zeta}{2} {\hat\phi}_n\right)\bar{\hat\iota}^{A'}\hat{o}_B \hat{o}_C...\hat{o}_F\cr
&&+(-1)^{n-k-1}\left( \hat{n}^a\hnabla_a {\hat\phi}_k - \hat{m}^a\hnabla_a{\hat\phi}_{k+1} + \frac{\cot\theta}{2\sqrt 2\sin\zeta}(n-2k-2) {\hat\phi}_{k+1}\right.\cr
&&\hspace{5cm}\left. + (k+1)\frac{\cos\zeta}{2} {\hat\phi}_k\right)\bar{\hat o}^{A'}\hat{o}_B \underbrace{\hat{o}_C...\hat{o}_M}_{k-1\, terms} \underbrace{\hat{\iota}_N...\hat{\iota}_F}_{n-k-1 \, terms}\cr
&&-(-1)^{n-k-1}\left( \hat{l}^a\hnabla_a {\hat\phi}_k - \bar{\hat m}^a\hnabla_a\hat{\phi}_{k-1} - \frac{\cot\theta}{2\sqrt 2\sin\zeta}(n-2k+2){\hat\phi}_{k-1}\right.\cr
&&\hspace{5cm}\left. - (n-k+1)\frac{\cos\zeta}{2} {\hat\phi}_k \right) \bar{\hat \iota}^{A'}\hat{\iota}_B \underbrace{\hat{o}_C...\hat{o}_M}_{k-1\, terms}\underbrace{\hat{\iota}_N...\hat{\iota}_F}_{n-k-1 \, terms}\cr
&&+ (-1)^n\left( \hat{n}^a\hnabla_a {\hat\phi}_0 - \hat{m}^a\hnabla_a {\hat\phi}_1 + \frac{\cot\theta}{2\sqrt 2\sin\zeta}(n-2)\frac{\cot\theta}{2} {\hat\phi}_1 + \frac{\cos\zeta}{2} {\hat\phi}_0 \right) \bar{\hat o}^{A'}\hat{\iota}_B\hat{\iota}_C...\hat{\iota}_F,
\end{eqnarray}
where $\underbrace{\hat{o}_C...\hat{o}_M}_{k-1\, terms} \underbrace{\hat{\iota}_N...\hat{\iota}_F}_{n-k-1 \, terms}\, (1\leq k \leq n-1)$ denotes the sum of $C^{k-1}_{n-2}$ components, where each component consists $(k-1)$  omicrons $\hat{o}_M$ and $(n-k-1)$ iotas $\hat{\iota}_N$ with $M,N\in \left\{C,D...F\right\}$. Multiplying \eqref{P} by $\bar{\hat\iota}_{A'}\hat{\iota}^B + \bar{\hat o}_{A'}\hat{o}^B$ and noting that $\bar{\hat o}_{A'}\bar{\hat o}^{A'} = \bar{\hat \iota}_{A'}\bar{\hat \iota}^{A'} = 0$, $\bar{\hat \iota}_{A'}\bar{\hat{o}}^{A'}= \hat{\iota}_B\hat{o}^B = -1$, $\bar{\hat o}_{A'}\bar{\hat{\iota}}^{A'} = \hat{o}_B\hat{\iota}^B=1$ we get
\begin{eqnarray}\label{PROEQ}
0 &=& \left(\bar{\hat\iota}_{A'}\hat{\iota}^B + \bar{\hat o}_{A'}\hat{o}^B \right)\hnabla^{AA'}\hat{\phi}_{AB...F} \cr
&=&  (-1)^{n-k-1}\left( (\hat{l}^a - \hat{n}^a)\hnabla_a {\hat\phi}_k + \hat{m}^a\hnabla_a{\hat\phi}_{k+1} - \bar{\hat m}^a\hnabla_a\hat{\phi}_{k-1} \right.\cr
&&\left.- \frac{\cot\theta}{2\sqrt 2\sin\zeta}(n-2k-2) {\hat\phi}_{k+1} - (k+1)\frac{\cos\zeta}{2} {\hat\phi}_k \right.\cr
&&\left.- \frac{\cot\theta}{2\sqrt 2\sin\zeta} (n-2k+2) {\hat\phi}_{k-1} - (n-k+1)\frac{\cos\zeta}{2} {\hat\phi}_k \right) \underbrace{\hat{o}_C...\hat{o}_M}_{k-1\, terms} \underbrace{\hat{\iota}_N...\hat{\iota}_F}_{n-k-1 \, terms}\cr
&=&(-1)^{n-k-1}\left[ {\sqrt 2}\partial_\zeta {\hat\phi}_k - \frac{1}{\sqrt 2\sin\zeta}\left( \partial_\theta - \frac{i}{\sin\theta}\partial_\varphi + (n-2k+2)\frac{\cot\theta}{2}  \right) {\hat\phi}_{k-1} \right.\cr 
&&\left. + \frac{1}{\sqrt 2\sin\zeta}\left( \partial_\theta + \frac{i}{\sin\theta}\partial_\varphi  - (n-2k-2)\frac{\cot\theta}{2}  \right) {\hat\phi}_{k+1} - (n+2)\frac{\cos\zeta}{2} {\hat\phi}_k \right]\underbrace{\hat{o}_C...\hat{o}_M}_{k-1\, terms} \underbrace{\hat{\iota}_N...\hat{\iota}_F}_{n-k-1 \, terms} .
\end{eqnarray}
Recall that $\hat{l}^a = \frac{1}{\sqrt{2}}(\partial_\tau + \partial_\zeta)$ and $\hat{n}^a = \frac{1}{\sqrt{2}}(\partial_\tau - \partial_\zeta)$, hence
$\hat{l}^a+\hat{n}^a = \sqrt{2}\partial_\tau = \mathcal{T}^a$. This leads to
$$\hat{o}^B\bar{\hat{o}}^{B'} + \hat{\iota}^B\bar{\hat\iota}^{B'} = \mathcal{T}^{BB'}.$$
Lowering two sides of this equality by $\bar{\hat{\varepsilon}}_{B'A'} = -\bar{\hat{o}}_{A'}\bar{\hat\iota}_{B'} + \bar{\hat\iota}_{A'}\bar{\hat o}_{B'}$ we get
$$\bar{\hat{o}}_{A'}\hat{o}^B + \bar{\hat\iota}_{A'}\hat{\iota}^B = \mathcal{T}^{B}_{A'}.$$
Plugging this into \eqref{PROEQ} we obtain that
\begin{eqnarray*}
0 &=& \mathcal{T}^{B}_{A'}\hnabla^{AA'}\hat{\phi}_{AB...F} \cr
&=&(-1)^{n-k-1}\left[ {\sqrt 2}\partial_\zeta {\hat\phi}_k - \frac{1}{\sqrt 2\sin\zeta}\left( \partial_\theta - \frac{i}{\sin\theta}\partial_\varphi + (n-2k+2)\frac{\cot\theta}{2}  \right) {\hat\phi}_{k-1} \right.\cr 
&&\left. + \frac{1}{\sqrt 2\sin\zeta}\left( \partial_\theta + \frac{i}{\sin\theta}\partial_\varphi  - (n-2k-2)\frac{\cot\theta}{2}  \right) {\hat\phi}_{k+1} - (n+2)\frac{\cos\zeta}{2} {\hat\phi}_k \right]\underbrace{\hat{o}_C...\hat{o}_M}_{k-1\, terms} \underbrace{\hat{\iota}_N...\hat{\iota}_F}_{n-k-1 \, terms}.
\end{eqnarray*}
From \eqref{InDer} we have that $\mathcal{T}^{B}_{A'}\hnabla^{AA'}\hat{\phi}_{AB...F} = \mathcal{T}^{AA'}\hnabla^B_{A'}\hat{\phi}_{AB...F}$. Therefore, the constraint system on the hypersurface $\left\{ \tau=constant \right\}$ in the term of spin components \eqref{CONE}:
\begin{gather*}
{\sqrt 2}\partial_\zeta {\hat\phi}_k - \frac{1}{\sqrt 2\sin\zeta}\left( \partial_\theta - \frac{i}{\sin\theta}\partial_\varphi + (n-2k+2)\frac{\cot\theta}{2}  \right) {\hat\phi}_{k-1} \cr 
+ \frac{1}{\sqrt 2\sin\zeta}\left( \partial_\theta + \frac{i}{\sin\theta}\partial_\varphi  - (n-2k-2)\frac{\cot\theta}{2}  \right) {\hat\phi}_{k+1} = (n+2)\frac{\cos\zeta}{2} {\hat\phi}_k \,\,(k=1,2...n-1)
\end{gather*}
corresponds to equation \eqref{con}:
$\left(\mathcal{T}^{AA'}\hnabla^B_{A'}\hat{\phi}_{AB...F}\right)|_{\tau=constant} =0$.

Now we prove that the constraint system has non-trivial solutions. As we have seen, the spin-$n/2$ zero rest-mass field equations 
\begin{equation}\label{zero_rest_mass}
\nabla^{AA'}\phi_{AB...F} = 0
\end{equation}
are an overdetermined system, which can be split into an evolution part and a spacelike constraint  that is preserved by the evolution. This constraint system is analogous to an elliptic equation, it is therefore not clear that it admits smooth compactly supported solutions. Penrose \cite{Pe1965} (see also recent \cite{ABJ}) shows that in Minkowski spacetime any solution of the spin-$n/2$ zero rest-mass field, at least locally can be obtained from a scalar potential (also called Hertz-type potential) $\chi$ satisfying the wave equation $\Box\chi = 0$. The construction is as follows : let $\phi_{(AB...F)}$ be a solution of the spin-$n/2$ zero rest-mass field equation and $\alpha_{A'}$ a spinor which is chosen to constant throughout Minkowski spacetime. Then at least locally we can find a spin-$(n-1)/2$ zero rest-mass field $\psi_{(B...F)}$ satisfying the spin-$(n-1)/2$ zero rest-mass equation $\nabla^{BB'}\psi_{B...F} = 0$ and
$$\phi_{AB...F} = \alpha_{A'}\nabla_A^{A'}\psi_{B...F} \, .$$
Continuing this process we get the scalar potential $\chi$ as follows
$$\phi_{AB...F} = \alpha_{A'}\beta_{B'}...\gamma_{F'}\nabla_{A}^{A'}\nabla_B^{B'}...\nabla_F^{F'}\chi$$
where $\chi$ satisfying $\Box\chi = 0$ and $\alpha_{A'},\beta_{B'}...\gamma_{F'}$ are given constant spinor.

Conversely we show that any solution $\chi$ satisfying the wave equation gives rise to a solution $\phi_{AB...F}$ of equation \eqref{zero_rest_mass} via a choice of constant spinors $\alpha_{A'},\beta_{B'}...\gamma_{F'}$. Indeed, we have
\begin{eqnarray*}
\nabla^{Z'A} \alpha_{A'}\beta_{B'}...\gamma_{F'}\nabla_{A}^{A'}\nabla_B^{B'}...\nabla_F^{F'}\chi &=& \alpha_{A'}\beta_{B'}...\gamma_{F'} \nabla^{Z'A}\nabla_{A}^{A'}\nabla_B^{B'}...\nabla_F^{F'}\chi \cr
&=& \alpha_{A'}\beta_{B'}...\gamma_{F'} \left( \frac{1}{2}\varepsilon^{Z'A'}\Box + \Box^{Z'A'}  \right) \nabla_B^{B'}...\nabla_F^{F'}\chi \cr
&=& \alpha_{A'}\beta_{B'}...\gamma_{F'}\frac{1}{2}\varepsilon^{Z'A'}  \nabla_B^{B'}...\nabla_F^{F'} \Box\chi \\
&=& 0.
\end{eqnarray*}
Here, due to our work on Minkowski spacetime, all the curvatures disappear, then $\Box^{Z'A'}  = 0$ and the wave operator can be commuted with the dervatives, for instance $[\Box,\nabla_B^{B'}] = 0.$

Therefore $\phi_{(AB...F)}: = \alpha_{A'}\beta_{B'}...\gamma_{F'}\nabla_{A}^{A'}\nabla_B^{B'}...\nabla_F^{F'}\chi$ is a solution of the equation \eqref{zero_rest_mass}. This shows that in Minkowski spacetime:
\begin{itemize}
\item[a)] If the energy of the initial data on $\Sigma_0$ is finite, then the energy of the solution on the space slices $\Sigma_T = \left\{t = T \right\}$ are also finite (by energy estimate). Since $\overline{\mathcal{C}_0^\infty(\Sigma_T)} = L^2(\Sigma_T)$, we can consider the wave equation $\Box \chi = 0$ with the compactly supported initial data, then we get a unique solution which is smooth compactly supported in space due to the finite propagation speed property. So that, there are solutions of spin $n/2$ zero rest-mass equation \eqref{zero_rest_mass} that are smooth and compactly supported in space.
\item[b)] As a consequence of $a)$, the constraint equations admit smooth compactly supported solutions and non-trivial finite energy solutions. It is not completely clear that all finite energy solutions can be obtained from a Hertz-type potential since the construction is merely local. Hence we shall work on the constrained subspace $\mathcal{H}_0 = \overline{\mathcal{C}_0^\infty (\Sigma_0,\mathbb{S}_{(AB...F)}) \cap \mathcal{D}}^{{L^2(\Sigma_0,\mathbb{S}_{AB...F})}}$.
\end{itemize}
Since the equation \eqref{zero_rest_mass} is conformally invariant, the same property is valid on the full and partial conformal compactification spacetimes $\hat{\mathbb M}$ and $\widetilde{\mathbb{M}}$.

\subsection{Generalisation of H\"ormander's result}\label{app_3_Goursat}
In this part we extend the result of H\"ormander \cite{Ho1990} for the spin-wave equations. The result of H\"ormander was extended for the scalar wave equation by Nicolas \cite{Ni2006} with the following minor modifications: the $\mathcal{C}^1$-metric, the continuous coefficients of the derivatives of the first order and the terms of order zero have locally $L^\infty$-coefficients. We refer \cite{Jo2020,Mo2019,Ni2016,Xuan2021} for the applications of the generalisation of H\"ormander's result to establish the wellposedness of the Goursat problem for the Dirac, Maxwell, linear and semilinear wave equations in the asymptotic simple and flat spacetimes. Here we will show that how the Goursat problem is valid for the spin-wave equations in the future $\mathcal{I}^+(\mathcal{S})$ of $\mathcal{S}$ in $\widetilde{{\mathbb M}}$ (recall that $\mathcal{S}$ is the spacelike hypersurface in $\widetilde{\mathbb M}$ such that it pass $\scri^+$ strictly in the past of the support of the initial data). 

\begin{figure}[H]
\begin{center}
\includegraphics[scale=0.5]{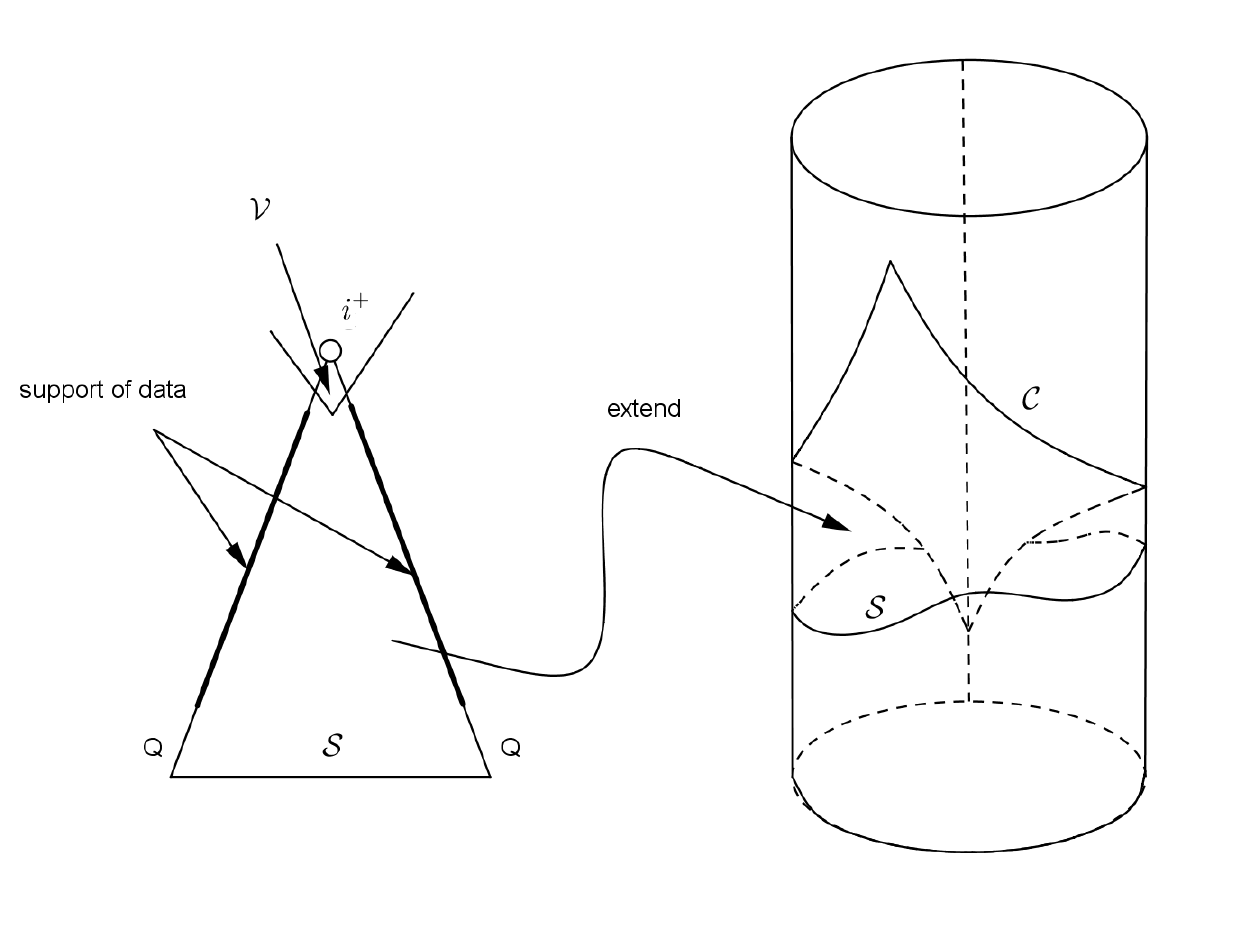}
\caption{The cutting off and extension of the future $\mathcal{I}^+(\mathcal{S})$.}
\end{center}
\end{figure}

Let $P$ be a point in the future $\mathcal{I}^+(\mathcal{S})$ in $\widetilde{\mathbb M}$, we cut off the future $\mathcal{V}$ of $P$ such that $\mathcal{V}$ does not contain the support of Goursat data and get the spacetime $\mathfrak{M}$. Then we extend $\mathfrak{M}$ onto a cylindrical globally hyperbolic spacetime $(\mathbb{R}_t\times S^3, \mathfrak{g})$, where $\mathfrak{g}|_{\hat{\mathfrak{M}}} = \tilde{g}|_{\mathfrak{M}}$ and the part of $\scri^+$ inside $\mathcal{I}^+(\mathcal{S})/\mathcal{V}$ is extended as a null hypersurface $\mathcal{C}$ (that is the graph of a Lipschitz function over $S^3$ and the data by zero on the rest of the extended hypersurface).

We consider the Goursat problem of the following spin wave equation in the spacetime $(\mathcal{M} = \mathbb{R}_t\times S^3, \mathfrak{g})$:
\begin{align}\label{spinwaveeq}
\begin{cases}
\Box_{\mathfrak{g}} \phi_{AB...F} &= 0,\cr
\phi_{AB...F}|_{\mathcal{C}} &= \psi_{AB...F}|_{\mathcal{C}} \in \mathcal{C}^\infty_0(\mathcal{C}, {\mathbb S}_{(AB...F)}),\cr
\nabla^{AA'}\phi_{AB...F}|_{\mathcal{C}} & = \zeta_{AB...F}|_{\mathcal{C}} \in \mathcal{C}^\infty_0(\mathcal{C}, {\mathbb S}_{(AB...F)}).
\end{cases}
\end{align}

Following \cite{Sti1936}, the spacetime $(\mathcal{M} = \mathbb{R}_t\times S^3, \mathfrak{g})$ is parallelizable, i.e, it admit a continuous global frame in the sense that the tangent space at each point has a basis. Therefore, we can chose a global spin-frame $\left\{o,\iota \right\}$ for $\mathcal{M}$ such that in this spin-frame the Newman-Penrose tetrad is $\mathcal {C}^{\infty}$.
Projecting \eqref{spinwaveeq} on $\left\{o,\iota \right\}$ (see the last of Appendix \ref{app_3_express} for the projection of the constrain equation $\nabla^{AA'}\phi_{AB...F}|_{\mathcal{C}}$) we get the scalar matrix form as follows
\begin{align}\label{system-wave}
\begin{cases}
P\Phi + L_1 \Phi &= 0,\cr
(\Phi, \,\partial_t\Phi)|_{t=0} &= (\Psi, \, \partial_t\Psi) \in \mathcal{C}^\infty_0(\mathcal{C})\times \mathcal{C}^\infty_0(\mathcal{C}), 
\end{cases}
\end{align}
where 
$$P= \left(\begin{matrix}
\Box&&0&&...&&0\\
0&&\Box&&...&&0\\
...&&...&&...&&...\\
0&&0&&...&&\Box
\end{matrix} \right)$$
is the $(n+1)\times (n+1)-$matrix diagram,
$$\Phi = \left( \begin{matrix}\phi_0\\ \phi_1\\...\\ \phi_n \end{matrix} \right), \, \Psi = \left( \begin{matrix}\psi_0\\ \psi_1\\...\\ \psi_n \end{matrix} \right)$$
is the components of $\phi_{AB...F}$ and $\Psi_{AB...F}$ respectively on the spin-frame $\left\{o,\iota \right\}$ and 
$$L_1 = \left( \begin{matrix} L_1^{00}&&L_1^{01}&&...&& L_1^{0n}\\ L_1^{10}&& L_1^{11}&&...&& L_1^{1n}\\...&&...&&...&&...\\ L_1^{n0}&& L_1^{n1}&&...&& L_1^{nn}  \end{matrix}  \right)$$
is the $(n+1)\times (n+1)-$ matrix where the components are the operators that have the coefficients $\mathcal {C}^\infty$:
$$L_1^{ij} = b_0^{ij}\partial_t + b_\alpha^{ij}\partial_\alpha + c^{ij}.$$

Since $\mathfrak{g}$ is a $\mathcal{C}^1$-metric, the first order terms in $L_1$ have continuous coefficients and the terms of order $0$ have locally $L^\infty$-coefficients, the Goursat problem for the $(n+1)\times (n+1)$-matrix wave equation \eqref{system-wave} is well-posed in $(\mathbb{R}_t\times S^3,\mathfrak{g})$ by applying theorems 3 and 4 in \cite{Ni2006}.
\begin{thm}
For the initial data $(\psi_i,\partial_t\psi_i) \in {\mathcal C}_0^{\infty}(\mathcal{C})\times {\mathcal{C}}_0^\infty(\mathcal{C}) $ for all $i=1,2...n$, the $(n+1)\times (n+1)$-matrix equation \eqref{system-wave}, hence the spin wave equation \eqref{spinwaveeq} has a unique solution $\Phi = (\phi_1,\phi_2...,\phi_n)$ satisfies
$$\phi_i \in {\mathcal {C}}(\mathbb{R};H^1(S^3))\cap \mathcal{C}^1(\mathbb{R};L^2(S^3)) \; \mbox{for all} \; i = 1,2...n.$$
\end{thm}
By using local uniqueness, causality and the finite propagation speed we have that the solution $\Phi$ vanishes in $\mathcal{I}^+(\mathcal{S})/\mathfrak{M}$. Therefore, the Goursat problem has a unique smooth solution in the future of $\mathcal{S}$, that is the restriction of $\Phi$ to $\mathfrak{M}$.

\subsection{Detailed calculations for the Goursat problem}\label{app_3_express}
We have the expression of the spinor field $\Xi^{A'}{_{\underbrace{(B...F)}_{n-1 \; indexs}}}$ on the rescaled spin-frame $\left\{\widetilde{o}_A,\,\widetilde{\iota}_A \right\}$ and its dual conjugation $\left\{\bar{\widetilde{o}}^{A'},\,\bar{\widetilde{\iota}}^{A'} \right\}$:
\begin{eqnarray*}
\Xi^{K'}{_{(B...F)}} &=& \sum_{k=0}^{n-1}(-1)^k \Xi^{1'}{_{n-1-k}} \bar{\widetilde{o}}^{K'} \underbrace{\widetilde{ \iota}_B...\widetilde{\iota}_c}_{k \; terms} \underbrace{\widetilde{o}_D...\widetilde{o}_F}_{n-1-k \; terms} \cr
&&- \sum_{k=0}^{n-1}(-1)^k \Xi^{0'}{_{n-1-k}} \bar{\widetilde{\iota}}^{K'} \underbrace{\widetilde{\iota}_{B}...\widetilde{\iota}_C}_{k \; terms} \underbrace{\widetilde{o}_D...\widetilde{o}_F}_{n-1-k \; terms}.
\end{eqnarray*}
The covariant derivative $\widetilde{\nabla}_{AK'}$ acts on the full spinor field can be decomposed as follows
\begin{eqnarray}\label{decompSpin}
\widetilde{\nabla}_{AK'} \Xi^{K'}{_{(B...F)}} &=& -(\widetilde{D}\Xi^{K'}{_{(B...F)}})\widetilde{\iota}_A\bar{\widetilde{\iota}}_{K'} + (\widetilde{D}'\Xi^{K'}{_{(B...F)}})\widetilde{o}_A\bar{\widetilde{o}}_{K'} \cr
&&- (\widetilde{\delta}\Xi^{K'}{_{(B...F)}})\widetilde{\iota}_A\bar{\widetilde{o}}_{K'} + (\widetilde{\delta}'\Xi^{K'}{_{(B...F)}})\widetilde{o}_A\bar{\widetilde{\iota}}_{K'}.
\end{eqnarray}
In the partial conformal compactification $\widetilde{M}$, we have the twelve values of the spin coefficients which are
$$\widetilde{\kappa} = \widetilde{\sigma} = \widetilde{\lambda} = \widetilde{\tau} = \widetilde{\nu} = \widetilde{\pi} = \widetilde{\rho} = \widetilde{\mu} = \widetilde{\epsilon} = 0,$$
$$\widetilde{\gamma} = -\frac{R}{\sqrt 2} \; , \; \widetilde{\beta} = - \widetilde{\alpha} = \frac{\cot\theta}{2\sqrt 2}.$$
The covariant derivatives act on the spin-frame $\left\{ \widetilde{o}_A,\, \widetilde{\iota}_A \right\}$ as (see \cite[eq. (4.5.26), pp. 227, Vol. 1]{PeRi}):
\begin{equation}\label{Modu1}
\begin{gathered}
\widetilde{D}\widetilde{o}_A = \widetilde{\varepsilon} \widetilde{o}_A - \widetilde{\kappa}\widetilde{\iota}_A = 0, \; \widetilde{D}\widetilde{\iota}_A = -\widetilde{\varepsilon}\widetilde{\iota}_A + \widetilde{\pi}\widetilde{o}_A = 0,\cr
\widetilde{\delta}'\widetilde{o}_A = \widetilde{\alpha} \widetilde{o}_A - \widetilde{\rho}\widetilde{\iota}_A   = -\frac{\cot\theta}{2\sqrt 2} \widetilde{o}_A, \; \widetilde{\delta}'\widetilde{\iota}_A = -\widetilde{\alpha}\widetilde{\iota}_A + \widetilde{\lambda}\widetilde{o}_A  = \frac{\cot\theta}{2\sqrt 2} \widetilde{\iota}_A,\cr
\widetilde{\delta}\widetilde{o}_A = \widetilde{\beta} \widetilde{o}_A - \widetilde{\sigma}\widetilde{\iota}_A   = \frac{\cot\theta}{2\sqrt 2} \widetilde{o}_A, \; \widetilde{\delta}'\widetilde{\iota}_A = -\widetilde{\beta}\widetilde{\iota}_A + \widetilde{\mu}\tilde{o}_A = -\frac{\cot\theta}{2\sqrt 2} \widetilde{\iota}_A,\cr
\widetilde{D}'\widetilde{o}_A = \widetilde{\gamma} \widetilde{o}_A - \widetilde{\tau}\widetilde{\iota}_A = -\frac{R}{\sqrt 2}\widetilde{o}_A, \; \widetilde{D}'\widetilde{\iota}_A = -\widetilde{\gamma}\widetilde{\iota}_A + \widetilde{\nu}\hat{\hat o}_A  = \frac{R}{\sqrt 2} \widetilde{\iota}_A.
\end{gathered}
\end{equation}
Similarly, on the dual conjugation spin-frame $\left\{ \bar{\widetilde{o}}^{A'}, \, \bar{\widetilde{\iota}}^{A'} \right\}$ we have
\begin{equation}\label{Modu2}
\begin{gathered}
\widetilde{D}\bar{\widetilde{o}}^{A'} = 0, \; \widetilde{D}\bar{\widetilde{\iota}}^{A'} = 0,\cr
\widetilde{\delta}'\bar{\widetilde{o}}^{A'} = -\frac{\cot\theta}{2\sqrt 2} \bar{\widetilde{o}}^{A'}, \; \widetilde{\delta}'\bar{\widetilde{\iota}}^{A'} = \frac{\cot\theta}{2\sqrt 2} \bar{\widetilde{\iota}}^{A'},\cr
\widetilde{\delta}\bar{\widetilde{o}}^{A'} = \frac{\cot\theta}{2\sqrt 2} \bar{\widetilde{o}}^{A'}, \; \widetilde{\delta}\bar{\widetilde{\iota}}^{A'} = -\frac{\cot\theta}{2\sqrt 2} \bar{\widetilde{\iota}}^{A'},\cr
\widetilde{D}'\bar{\widetilde{o}}^{A'} = -\frac{R}{\sqrt 2}\bar{\widetilde{o}}^{A'}, \; \widetilde{D}'\bar{\widetilde{\iota}}^{A'} = \frac{R}{\sqrt 2} \bar{\widetilde{\iota}}^{A'}.
\end{gathered}
\end{equation}
Therefore, since \eqref{decompSpin}, \eqref{Modu1} and \eqref{Modu2}, we obtain the detailed expression of $\widetilde{\nabla}_{AK'}\Xi^{K'}{_{(B...F)}}$ as 
\begin{eqnarray*}
&&\widetilde{\nabla}_{AK'}\Xi^{K'}{_{B...F}} = (\widetilde{D}\Xi^{K'}{_{B...F}})\widetilde{\iota}_A\bar{\widetilde{\iota}}_{K'} + (\widetilde{D}'\Xi^{K'}{_{B...F}})\widetilde{o}_A \bar{\widetilde{o}}_{K'} - (\widetilde{\delta}\Xi^{K'}{_{B...F}})\widetilde{\iota}_{A} \bar{\widetilde{o}}_{K'} - (\widetilde{\delta}'\Xi^{K'}{_{B...F}})\widetilde{o}_A \bar{\widetilde{\iota}}_{K'} \cr
&=& -\widetilde{D} \left( \sum_{k=0}^{n-1}(-1)^k \Xi^{1'}{_{n-1-k}} \bar{\widetilde{o}}^{K'} \underbrace{\widetilde{\iota}_B...\widetilde{\iota}_C}_{k \; terms} \underbrace{\widetilde{o}_D...\widetilde{o}_F}_{n-1-k \; terms} \right) \widetilde{\iota}_A \bar{\widetilde{\iota}}_{K'} \cr
&& + \widetilde{D}' \left( \sum_{k=0}^{n-1}(-1)^k \Xi^{0'}{_{n-1-k}} \bar{\widetilde{\iota}}^{K'} \underbrace{\widetilde{\iota}_{B}...\widetilde{\iota}_C}_{k \; terms} \underbrace{\widetilde{o}_D...\widetilde{o}_F}_{n-1-k \; terms} \right) \widetilde{o}_A \bar{\widetilde{o}}_{K'}\cr
&& - \widetilde{\delta} \left(\sum_{k=0}^{n-1}(-1)^k \Xi^{0'}{_{n-1-k}} \bar{\widetilde{\iota}}^{K'} \underbrace{\widetilde{\iota}_{B}...\widetilde{\iota}_C}_{k \; terms} \underbrace{\widetilde{o}_D...\widetilde{o}_F}_{n-1-k \; terms}  \right) \widetilde{\iota}_{A} \bar{\widetilde{o}}_{K'} \cr
&& + \widetilde{\delta}' \left( \sum_{k=0}^{n-1}(-1)^k \Xi^{1'}{_{n-1-k}} \bar{\widetilde{o}}^{K'} \underbrace{\widetilde{ \iota}_B...\widetilde{\iota}_C}_{k \; terms} \underbrace{\widetilde{o}_D...\widetilde{o}_F}_{n-1-k \; terms}  \right) \widetilde{o}_A \bar{\widetilde{\iota}}_{K'} \cr
&=& \sum_{k=0}^{n-1} (-1)^k \left( \widetilde{D} \Xi^{1'}{_{n-1-k}} - \widetilde{\delta} \Xi^{0'}{_{n-1-k}} - \frac{n-2-2k}{2\sqrt 2}\cot\theta \Xi^{0'}{_{n-1-k}}\right) \widetilde{\iota}_A\underbrace{\widetilde{\iota}_B...\widetilde{\iota}_C}_{k \; terms} \underbrace{\widetilde{o}_D...\widetilde{o}_F}_{n-1-k \; terms} \cr
&&+ \sum_{k=0}^{n-1} (-1)^k \left\{ \left( \widetilde{D}' - \frac{(2k+2-n)R}{\sqrt 2} \right)\Xi^{0'}{_{n-1-k}} \right.\cr
&&\hspace{6cm}\left. - \left( \widetilde{\delta}' + \frac{2k - n}{2\sqrt 2}\cot\theta \right)\Xi^{1'}{_{n-1-k}}  \right\}  \widetilde{o}_A \underbrace{\widetilde{\iota}_B...\widetilde{\iota}_C}_{k \; terms} \underbrace{\widetilde{o}_D...\widetilde{o}_F}_{n-1-k \; terms} .
\end{eqnarray*}
Therefore, the equation $\widetilde{\nabla}_{AK'}\Xi^{K'}{_{B...F}} = 0$ is equivalent to the following system
\begin{align*}
\begin{cases}
\widetilde{D} \Xi^{1'}{_{n-1-k}} - \widetilde{\delta} \Xi^{0'}{_{n-1-k}} - \frac{n-2-2k}{2\sqrt 2}\cot\theta \Xi^{0'}{_{n-1-k}} &= 0 \, ,\cr
\left(\widetilde{D}' - \frac{(2k+2-n)R}{\sqrt 2} \right)\Xi^{0'}{_{n-1-k}} - \left( \widetilde{\delta}' + \frac{2k - n}{2\sqrt 2}\cot\theta \right)\Xi^{1'}{_{n-1-k}} &= 0  
\end{cases}
\end{align*}
for all $k=0,1...n-1$. Taking the constrain of this system on $\scri^+$, we get only the constraint of the second equations
$$\sqrt 2 \partial_u \Xi^{0'}{_{n-1-k}}|_{\scri^+} = 0,$$
for all $k=0,1...n-1$. Integrating these equations along $\scri^+$, we get $\Xi^{0'}{_{n-1-k}}|_{\scri^+} = \; constant$. This leads to a fact that the Cauchy problem with the initial condition $\Xi^{0'}{_{n-1-k}}|_{{\mathcal V}(P)} = 0$ has a unique solution and it equals to zero.

\end{document}